\documentclass[12pt]{article}
\usepackage[margin=1in]{geometry}
\usepackage{hyperref}
\hypersetup{colorlinks, citecolor=blue, filecolor=blue, linkcolor=blue, urlcolor=blue}
\usepackage{graphicx}
\usepackage{url}
\usepackage[longnamesfirst]{natbib}
\usepackage{amsmath, amssymb, amsthm}
\usepackage{bbm}
\usepackage{mathtools}
\usepackage{etoolbox} 
\usepackage[title]{appendix}
\usepackage{engord}
\usepackage{float}
\usepackage{pdflscape}
\usepackage{booktabs}
\usepackage[noabbrev, nameinlink]{cleveref}
\usepackage{calc}
\usepackage{tikz, pgfplots}
\usetikzlibrary{calc}
\usetikzlibrary{arrows.meta}
\usetikzlibrary{decorations.pathreplacing}
\usepgfplotslibrary{fillbetween}
\pgfplotsset{compat=1.14}
\pgfplotsset{every axis label/.append style={font=\tiny}}
\usepackage[labelsep=period]{caption}
\usepackage{subcaption}
\usepackage{multirow}
\usepackage{xr}
\usepackage{setspace}
\usepackage{sectsty}
\usepackage[most]{tcolorbox}

\usepackage[most]{tcolorbox}
\usepackage{titlesec}

\newtcolorbox{mechanismbox}[1]{
    enhanced,
    breakable,
    colback=gray!2,
    colframe=gray!2,
    borderline west={1.2pt}{0pt}{black!55},
    boxrule=0pt,
    arc=0pt,
    left=9pt,
    right=6pt,
    top=6pt,
    bottom=6pt,
    title={#1},
    fonttitle=\bfseries,
    coltitle=black,
    attach title to upper={\par\vspace{2pt}}
}
\sectionfont{\large}
\subsectionfont{\normalsize}
\subsubsectionfont{\normalsize}

\newcommand{\topref}[2]{\texorpdfstring{\Cref{#2}}{#1{}\ref{#2}}}

\ExplSyntaxOn
\newcounter{save}
\seq_new:N \g_ephraim_save_seq
\NewDocumentEnvironment{save}{o+b}
{
	\refstepcounter{save}
	\IfValueT{#1}{\label{save@#1}}
	\seq_gput_right:Nn\g_ephraim_save_seq{#2}
}{}
\NewDocumentCommand{\load}{m}
{
	\clist_map_inline:nn {#1}
	{
		\seq_item:Nn\g_ephraim_save_seq
		{
			\getrefnumber{save@##1}
		}
		\par
	}
}
\NewDocumentCommand{\loadall}{}
{
	\seq_use:Nn\g_ephraim_save_seq{\par}
}
\ExplSyntaxOff

\newcommand{\comma}{,}

\newcommand{\D}{\mathrm{d}}
\newcommand{\switch}[1]{\left\{\begin{aligned}#1\end{aligned}\right.}

\newcommand{\st}{\text{s.t. }}

\newcommand{\rep}[2]{
  \prg_replicate:nn {#2} {#1}%
}

\newcommand{\htt}{\hat{\theta}}
\newcommand{\utt}{\tau}
\newcommand{\Eta}{\mathrm{H}}

\allowdisplaybreaks

\numberwithin{equation}{section}

\newtheorem{lemma}{Lemma}

\newtheorem{proposition}{Proposition}
\newtheorem{theorem}{Theorem}

\newtheorem{corollary}{Corollary}[theorem]

\newtheorem*{claim*}{Claim}
\newtheorem*{literature*}{Related literature}

\title{Coordinating Treatment Allocation and Recommendation\thanks{We thank Cuimin Ba, Aislinn Bohren, Kevin He, Navin Kartik, Xiao Lin, George Mailath, Margaret Meyer, Stephen Morris, Alessandro Pavan, Andrew Postlewaite, Daniel Rappoport, Doron Ravid, Collin Raymond, Juuso Toikka, Maren Vairo, and Rakesh Vohra, as well as participants at the Penn Micro Theory Lunch, the Midwest Trade \& Theory Conference, the Pennsylvania Economics Theory Conference (PETCO), the ACM Conference on Economics and Computation (EC'26), the 37th Stony Brook International Conference on Game Theory, and the North American Summer Meeting of the Econometric Society for valuable comments and suggestions.
}}

\author{Li Guo\thanks{University of Pennsylvania.
\href{mailto:gary16@sas.upenn.edu}{gary16@sas.upenn.edu}} \and Penghuan Yan\thanks{University of Pennsylvania.  \href{mailto:yanph@sas.upenn.edu}{yanph@sas.upenn.edu}} }

\date{ \vspace*{0.5cm} \today} 

\begin{document}

\bgroup
\let\footnoterule\relax

\begin{singlespace}
\maketitle

\begin{abstract}
    \noindent We study a model in which a sender allocates limited treatment to agents with heterogeneous quality and later recommends selected agents to a receiver, seeking to maximize the number of agents accepted by the receiver. All agents value treatment, which improves agents' quality, but treatment must be allocated before the sender observes agents' initial quality; recommendation occurs only after quality is learned. A natural benchmark is to design the two instruments separately: allocate treatment randomly first, and then recommend agents from the top down afterward. Our main result shows that the sender can do strictly better by coordinating treatment allocation with recommendations. In the optimal joint mechanism, treatment is non-monotone in quality: an intermediate group has a lower treatment probability than both higher- and lower-quality agents, but is compensated with a guaranteed recommendation when treatment is realized. We provide an implementation through contracts that induce self-selection and discuss applications to education, industrial policy, and startup incubation. 
\end{abstract}
\end{singlespace}
\thispagestyle{empty}

\clearpage
\egroup
\setcounter{page}{1}



\section{Introduction}

Many institutions both expend resources to improve agents' quality and recommend selected agents to outside evaluators. Schools provide students with educational resources and later recommend them to employers for hiring. Incubators support startups and later recommend selected startups to investors who decide whether to invest in them. We refer to such an institution as the sender. In these settings, the sender must make two choices: how to initially allocate limited treatment that improves the treated agents' quality, and how to later recommend agents to an outside receiver.

All agents value treatment, but it is scarce and must be allocated before the sender observes agents' private types. A school may allocate educational resources before knowing which students will become strong candidates. An incubator may allocate support before knowing which startups have real potential. Only after treatment is allocated does the sender learn more: schools observe course performance, and incubators observe demand data. The sender can then recommend selected agents to the outside receiver.

In this paper, we focus on how a sender should allocate limited treatment and design recommendations to maximize the number of agents accepted by the receiver. To answer this question, we study a joint mechanism and information design problem that involves a sender, a receiver, and a continuum of agents with heterogeneous quality. The sender wants to maximize the mass of agents accepted by the receiver. The sender commits ex ante to how treatment and recommendation depend on agents' reports and subsequent information. First, before observing agents' true types, treatment is allocated to a limited mass of agents. Treatment improves treated agents' quality. Second, after treatment is realized and agents' types are observed, the committed rule determines which agents are recommended. The receiver observes only the sender's recommendation and accepts recommended agents only if their expected quality is sufficiently high. Agents value both treatment and acceptance, but value acceptance more.

A natural benchmark is to design the two instruments separately. The sender first allocates treatment, and later, after observing agents' quality, decides whom to recommend. Since the sender cannot yet identify quality and all agents want treatment, treatment is assigned at random and each agent receives treatment with the same probability. After treatment is realized and quality is observed, the sender recommends agents in descending order of post-treatment quality. The recommendation cutoff is chosen so that the average quality of recommended agents is just high enough for the receiver to accept them.

Our main result shows that the sender can do strictly better by coordinating treatment allocation with the subsequent post-treatment recommendations. The optimal joint mechanism targets a group of intermediate-quality agents. These agents receive a lower probability of treatment than others, but are compensated with a recommendation as long as they are treated. In this way, the sender reduces the amount of treatment spent on this group while preserving their incentives to report truthfully. Consequently, some of the saved treatment opportunities go to high-quality agents, who are more likely to be recommended, thereby improving the overall quality of recommended agents.

The treatment allocation in the optimal joint mechanism is non-monotone. Treatment is useful to the sender only when it reaches agents who are eventually recommended. The sender therefore wants to reserve treatment for high-quality agents, who are more likely to enter the recommended pool. The mechanism instead offers a special arrangement to an intermediate group: these agents receive a lower treatment probability and are compensated with treatment-contingent recommendation. This saves capacity and redirects treatment toward higher-quality agents. At the same time, the sender does not want the lowest-quality agents to choose this arrangement, because recommending them would lower the average quality of the recommended pool too much. To separate them from the intermediate group, the mechanism gives the lowest-quality agents a higher treatment probability than the intermediate group, while giving them no chance of recommendation.

In standard information design, a recommendation is mainly a signal sent to the receiver. Here, recommendation also serves as an incentive device for agents because the sender can commit to the recommendation rule before treatment is allocated. The same instrument therefore performs two tasks: it persuades the receiver after quality is learned, and it screens agents before treatment is assigned.

We provide a simple implementation of the optimal joint mechanism using a menu with two contracts that induces self-selection. Accepting the non-default contract means receiving a lower probability of treatment, but facing a more lenient standard for recommendation when treated. In equilibrium, only the targeted group accepts the non-default contract; all other agents choose the default option and face the regular treatment and recommendation rule. 

A natural application is education. The school can implement the optimal mechanism by offering two tracks, such as a regular track and an honors track. The honors track may provide better access to scarce educational resources, such as advanced courses or research opportunities, but it may also impose more demanding standards.\footnote{Similar track structures already exist in practice. For example, some economics master's programs offer a research-intensive track that gives students greater access to PhD-level coursework while also imposing higher requirements.} We also discuss applications to industrial policy and startup incubation.

\subsection*{Related Literature}
Our paper builds on Bayesian persuasion and information design. In this literature, a sender designs signals to influence a receiver's beliefs and actions \citep{RayoSegal2010, KamenicaGentzkow2011, BergemannMorris2016, Kolotilin2018, BergemannMorris2019, Kamenica2019}. Our paper has the same persuasion motive: the sender uses information to induce acceptance by a receiver. The difference is that the payoff-relevant state is endogenous: the quality distribution entering the recommendation problem depends on the sender's treatment-allocation rule.

The treatment-allocation part of the model relates to mechanism design and screening. In this literature, a designer allocates goods or contracts among privately informed agents subject to incentive constraints \citep{MussaRosen1978, Myerson1981, BaronMyerson1982, MaskinRiley1984, LaffontMartimort2002}. Our model is especially related to mechanism design without monetary transfers \citep{MartimortSemenov2006, SchummerVohra2007, Miralles2012, ben2014}. In our model, the only instrument through which the sender can transfer utility to agents after treatment allocation is the later recommendation. At the same time, the sender's objective is to use this recommendation to induce receiver acceptance. This dual role of recommendation is the key difference: the same recommendation rule must provide incentives to agents before treatment is allocated and remain credible to the receiver afterward.

The closest strand of literature studies the joint design of information and mechanisms \citep{BergemannPesendorfer2007, kolotilin2017persuasion, BergemannHeumannMorris2026a}. These papers share with ours the idea that information and mechanisms should be designed jointly. In much of this literature, information design and mechanism design are directed at the same side of the market. Our environment is three-sided. The sender designs a treatment-allocation mechanism for agents, but designs information for a separate receiver. The recommendation connects the two sides: it is the sender's only instrument for transferring utility to agents, and it is also the information device used to induce receiver acceptance. Two closely related papers are \citet{chopra2025incentive} and \citet{Dworczak2020}. \citet{chopra2025incentive} study a three-sided setting in which a profit-seeking intermediary designs and sells information products, such as experiments, tests, or certificates, to privately informed senders, who then use the purchased information to persuade receivers. \citet{Dworczak2020} studies a setting in which a designer first runs a payoff-relevant mechanism and then chooses whether and how to disclose information elicited by that mechanism to a third party. The key difference is the role of the mechanism-design stage. In these papers, the mechanism-design stage either generates revenue from selling information or directly determines a payoff-relevant allocation. In our model, treatment allocation is not itself the sender's payoff-relevant objective. Instead, the mechanism allocates scarce treatment to target agents, allowing the sender to recommend more agents to the receiver.

Our paper is also related to work on certification, ratings, grading, testing, and information as an incentive device. Classic work on certification and quality disclosure studies how intermediaries or firms disclose quality information to markets \citep{Lizzeri1999, DranoveJin2010}. More recent work studies how certification, ratings, grades, or tests can induce agents to take costly actions that improve quality, such as effort, investment, or input choices \citep{zubrickas2015optimal, BoleslavskyKim2018, Zapechelnyuk2020, Xiao2024, CamboniCarnehl2025, SaeediShourideh2026}. In these papers, the analogue of treatment is typically costly for agents: it lowers their utility but raises the quality later revealed or rewarded by certification, ratings, grades, or tests. We study a complementary case in which treatment also raises quality but is valued by agents. The sender therefore does not use recommendation to induce agents to take a costly action; it uses recommendation to screen agents and improve the allocation of a scarce treatment.

The rest of the paper is organized as follows. \Cref{sec:model} sets up the model. \Cref{sec:optimal-direct} characterizes the optimal joint mechanism, develops the main intuition, and provides an implementation of the mechanism via contracts that induce self-selection. \Cref{sec:applications} discusses three applications: education, industrial policy, and startup incubation. \Cref{sec:conclusion} concludes.

\section{Model}
\label{sec:model}

\subsection{Environment}

There is a sender, a receiver, and a unit mass of agents. The sender has two instruments, both specified ex ante by a committed mechanism: an allocation rule that distributes scarce treatment, which improves treated agents' quality, and a recommendation rule that generates signals to the receiver. The sender's objective is to maximize the mass of agents accepted by the receiver. 

Each agent has a type \(\theta\in\Theta=[0,1]\), drawn from a distribution \(F\) with positive density \(f\). The type \(\theta\) represents the agent's initial quality. Agents privately observe their types before treatment is assigned. The sender does not observe agents' types at the treatment-allocation stage, but observes them before recommendation signals are sent to the receiver. The receiver observes neither agents' types nor their treatment statuses.

Treatment is binary. Let \(d\in\{0,1\}\) denote an agent's treatment status. If a type-\(\theta\) agent is treated, his quality becomes \(\theta+v\); otherwise, his quality remains \(\theta\). The constant \(v>0\) is the treatment effect. Equivalently, final quality is \(\theta+vd\). Treatment capacity is limited: the sender can treat at most a mass \(\kappa\in(0,1)\) of agents.

The sender commits ex ante to a treatment-allocation and recommendation mechanism. Let \(\mathcal M\) denote the message space from agents to the mechanism, and let \(\mathcal S\) denote the recommendation-signal space from the sender to the receiver. Given \(\mathcal M\) and \(\mathcal S\), a mechanism is denoted by $\pi=(p,q)$.

The first component is a treatment-allocation rule
\[
    p:\mathcal M\to\Delta\{0,1\},
\]
where \(p(d\mid m)\) is the probability that an agent receives treatment $d$ after sending message \(m\in\mathcal M\). Since treatment decision is binary, we write $p(m)=p(1\mid m)$. The second component is a recommendation rule
\[
    q:\mathcal M\times\Theta\times\{0,1\}\to\Delta(\mathcal S),
\]
where \(q(s\mid m,\theta,d)\) is the probability that recommendation signal \(s\in\mathcal S\) is sent to the receiver after the agent has sent message \(m\), the sender has observed the agent's type \(\theta\), and treatment status \(d\) has been realized.

The receiver observes only the sender's recommendation signal and decides whether to accept the agent. A signal can induce acceptance only if the posterior expected final quality conditional on that signal is sufficiently high. That is, after observing signal $s$, the receiver accepts only if
\[
\mathbb E_\pi[\theta+vd\mid s]\geq \tau,
\]
where the expectation is taken under the distribution induced by the mechanism and agents' messages. The sender's payoff is the mass of agents accepted by the receiver.

Agents value both treatment and acceptance. The payoff from being accepted is normalized to \(1\). Each agent obtains utility \(u\) from receiving treatment, where \(0<u<1\). Thus, acceptance is more valuable than treatment, but all agents strictly prefer treatment to no treatment.

The timing is as follows. First, the sender commits to a mechanism \(\pi=(p,q)\). Second, each agent observes his type and sends a message \(m\in\mathcal M\). Third, treatment is assigned according to the treatment-allocation rule \(p\), and treatment status is realized. Fourth, the agent's true type and treatment status are observed by the sender, and a recommendation signal is sent to the receiver according to the recommendation rule \(q(\cdot\mid m,\theta,d)\). Finally, the receiver decides whether to accept the agent.

\subsection{Direct Mechanisms}

By the revelation principle, we can restrict attention to direct mechanisms. In a direct mechanism, the agent-message space is \(\mathcal M=\Theta\). Each agent reports a type \(\hat{\theta}\in\Theta\), and truth-telling requires \(\hat{\theta}=\theta\) on the equilibrium path.

A direct mechanism is denoted by \(\pi=(p,q)\). The first component, \(p\), is a treatment-allocation rule,
\[
    p:\Theta\to[0,1],
\]
where \(p(\hat{\theta})\) is the probability that an agent receives treatment after reporting \(\hat{\theta}\).

The second component, \(q\), is a recommendation rule,
\[
    q:\Theta\times\Theta\times\{0,1\}\to \Delta\mathcal S,
\]
where \(q(s\mid \hat{\theta},\theta,d)\) is the probability that recommendation signal \(s\in\mathcal S\) is sent to the receiver after the agent has reported \(\hat{\theta}\), the sender has observed the agent's true type \(\theta\), and treatment status \(d\) has been realized.

Since the receiver's action is binary, we restrict attention without loss to a binary recommendation-signal space, \(\mathcal S=\{0,1\}\). We interpret \(s=1\) as a recommendation and \(s=0\) as no recommendation. With binary recommendation signals, we write
\[
    q(\hat{\theta},\theta,d)
    :=
    q(1\mid \hat{\theta},\theta,d)
\]
for the probability of recommendation. A recommendation is intended to induce acceptance by the receiver, subject to the obedience constraint below.

\subsection{Receiver Obedience and Agent Incentives}

Under truth-telling $\hat{\theta}=\theta$, the posterior mean quality of a recommended agent is
\begin{equation*}
    \mu_\pi
    :=\mathbb E_\pi[\theta+vd\mid s=1] =
    \frac{
    \displaystyle\int_{\Theta}
    [
        \overbrace{
        p(\theta)q(\theta,\theta,1)\bigl(\theta+v\bigr)
        }^{\text{treated quality}}
        +
        \overbrace{
        \bigl(1-p(\theta)\bigr)q(\theta,\theta,0) \ \theta
        }^{\text{untreated quality}}
    ]\D F(\theta)
    }{
    \underbrace{
    \int_{\Theta}
    \left[
        p(\theta)q(\theta,\theta,1)
        +
        \bigl(1-p(\theta)\bigr)q(\theta,\theta,0)
    \right]\D F(\theta)
    }_{\text{mass of recommended agents}}
    },
\end{equation*}
whenever the denominator is positive. This expression integrates out treatment status: the receiver does not observe whether a recommended agent was treated, but infers the composition of the recommended pool from the sender's committed rules.

Receiver obedience requires
\begin{equation*}
    \mu_\pi\geq \tau.
\end{equation*}

Given the agents' preferences, if a type-\(\theta\) agent reports \(\hat{\theta}\), his expected payoff is
\begin{equation*}
    U_\pi(\theta,\hat{\theta})
    =
    \underbrace{
    p(\hat{\theta})
    \left[
        u+q(\hat{\theta},\theta,1)
    \right]
    }_{\text{payoff when treated}}
    +
    \underbrace{
    \bigl(1-p(\hat{\theta})\bigr)
    q(\hat{\theta},\theta,0)
    }_{\text{payoff when untreated}} .
\end{equation*}
With probability \(p(\hat{\theta})\), the agent receives treatment and obtains treatment utility \(u\); conditional on being treated, he is recommended with probability \(q(\hat{\theta},\theta,1)\). With probability \(1-p(\hat{\theta})\), the agent is untreated and obtains only the acceptance payoff, which occurs with probability \(q(\hat{\theta},\theta,0)\).

Truthful reporting requires
\begin{equation*}
    U_\pi(\theta,\theta)\geq U_\pi(\theta,\hat{\theta})
    ,\quad
    \forall \theta,\hat{\theta}\in\Theta .
\end{equation*}

\subsection{Sender's problem}

The sender chooses a feasible direct mechanism to maximize the mass of recommended agents. Under truth-telling, the sender solves
\begin{equation}
\label{eq:sender-primal}
\begin{aligned}
    \max_{p,q}\quad
    &
    \int_{\Theta}
    \left[
        p(\theta)q(\theta,\theta,1)
        +
        \bigl(1-p(\theta)\bigr)q(\theta,\theta,0)
    \right]\D F(\theta), \\
    \text{s.t.}\quad
    & \int_{\Theta} p(\theta)\,\D F(\theta)\leq \kappa,
    && \text{(Capacity)} \\
    & \mu_\pi\geq \tau,
    && \text{(Obedience)} \\
    & U_\pi(\theta,\theta)\geq U_\pi(\theta,\hat{\theta}),
    \quad \forall \theta,\hat{\theta}\in\Theta.
    && \text{(IC)}
\end{aligned}
\end{equation}
The model combines a mechanism-design problem on the agent side with an information-design problem on the receiver side. Treatment is allocated before the sender observes agents' types, so the allocation rule must respect agents' reporting incentives. Recommendation determines the receiver's posterior belief, so the recommendation rule must satisfy receiver obedience.

\section{Optimal Direct Mechanism}
\label{sec:optimal-direct}

In this section, we characterize the optimal direct mechanism. We focus on the nondegenerate interior case in which recommending everyone violates the obedience constraint, i.e.
\[
    \utt>\int_\Theta \theta\,\D F(\theta)+v\kappa.
\]

Our main result is that the optimal direct mechanism takes a four-segment form.

\begin{theorem}
\label{thm:optimal-direct}
In the nondegenerate interior case, an optimal direct mechanism either satisfies $q(\hat\theta,\theta,d)=0$ almost everywhere (recommends no one), or takes the following form:
\begin{equation*}
    \bigl(p(\theta),q(\theta,\theta,1),q(\theta,\theta,0)\bigr)
    =
    \begin{cases}
    (P,1,1),
    & \theta\in(\theta_3,1],\\[0.4em]
    (P,1,0),
    & \theta\in(\theta_2,\theta_3],\\[0.4em]
    \left(\dfrac{u}{1+u}P,1,0\right),
    & \theta\in(\theta_1,\theta_2],\\[0.9em]
    (P,0,0),
    & \theta\in[0,\theta_1],
    \end{cases}
\end{equation*}
where \(\kappa<P\leq 1\) and $0\leq \theta_1<\theta_2\leq \theta_3\leq 1$ are cutoffs,
and
\begin{equation*}
    q(\hat{\theta},\theta,d)=0,
    \quad
    \text{for all }\hat{\theta}\neq\theta
    \text{ and }d\in\{0,1\}.
\end{equation*}
\end{theorem}
The mechanism divides agents into four groups. The top group receives the high treatment probability \(P\) and is recommended regardless of treatment status. The upper-middle group also receives the high treatment probability \(P\) and is recommended only if treated. The lower-middle group is the target group: these agents receive a lower treatment probability, \(\frac{u}{1+u}P\), and are recommended conditional on being treated. The bottom group receives the high treatment probability \(P\) but is never recommended.

To understand the intuition behind \Cref{thm:optimal-direct}, it is useful to compare the optimal mechanism with a separated-design benchmark.

\subsection{Separated-Design Benchmark}
\label{subsec:separate}
In this benchmark, treatment allocation and recommendation are designed separately. At the treatment-allocation stage, the sender only chooses which agents receive treatment. After treatment is realized and types are observed, the sender publicly commits to a recommendation rule that does not depend on any communication with agents in the treatment-allocation stage. The solution to this recommendation problem is characterized by a cutoff \(\theta^S\), and the benchmark mechanism is
\begin{equation*}
\begin{aligned}
    \bigl(p^S(\theta),q^S(\theta,\theta,1),q^S(\theta,\theta,0)\bigr)
    =
    \begin{cases}
    \bigl(P^S,1,1\bigr),
    & \theta^S<\theta\leq 1,\\[0.4em]
    \bigl(P^S,1,0\bigr),
    & \theta^S-v<\theta\leq\theta^S,\\[0.4em]
    \bigl(P^S,0,0\bigr),
    & 0 \leq \theta\leq\theta^S-v,
    \end{cases}
    \qquad
    P^S=\kappa .
\end{aligned}
\tag{S}
\label{eq:separated-design}
\end{equation*}
The cutoff \(\theta^S\) is chosen so that the receiver's obedience constraint binds.

The cutoff form has a simple interpretation. Since treatment allocation is designed separately from recommendation and all agents value treatment, before types are observed, the sender has no way to direct treatment toward particular types. Treatment is therefore allocated uniformly: \(P^S=\kappa\).

Once treatment is realized and types are observed, the sender faces a standard recommendation problem. Since treatment raises quality by \(v\), a treated agent of type \(\theta\) has quality \(\theta+v\), while an untreated agent has quality \(\theta\). The sender therefore recommends agents from the top down in realized quality. Untreated agents are recommended only if \(\theta>\theta^S\), while treated agents are recommended if \(\theta+v>\theta^S\), or equivalently \(\theta>\theta^S-v\). Thus high types are recommended regardless of treatment status, intermediate types are recommended only if treated, and low types are never recommended.

The limitation of the separated benchmark is that treatment cannot be targeted. Treatment helps the sender only when it is assigned to agents who are eventually recommended. If a treated agent is not recommended, the treatment uses capacity but does not improve the quality of the recommended pool. However, in the separated benchmark, treatment is forced to be assigned randomly, so some treatment is spent on agents who are too low-quality to be recommended even after treatment.

The sender would instead like treatment to reach agents with high initial quality, since these agents are more likely to enter the recommended pool. Treating such agents raises the quality of the recommended pool and relaxes the receiver's obedience constraint, which allows the sender to recommend more agents. The problem is that, at the treatment-allocation stage, the sender cannot distinguish these agents from lower-quality agents. The optimal joint mechanism in \Cref{thm:optimal-direct} improves on the separated benchmark by using the future recommendation rule to create incentives that help direct treatment toward the relevant types.

\subsection{Intuition for \topref{Theorem }{thm:optimal-direct}}
\label{subsec:intuition-theorem1}

The optimal joint mechanism differs from the separated benchmark by introducing a target group. In the separated benchmark, low-type agents are never recommended, regardless of treatment status, but they still receive the same treatment probability as all other agents. The optimal joint mechanism instead selects some relatively low types and gives them a lower treatment probability, while compensating them with a recommendation when treatment is realized. This is the lower-middle group,
\begin{equation*}
    \theta\in(\theta_1,\theta_2],
    \qquad
    \bigl(p(\theta),q(\theta,\theta,1),q(\theta,\theta,0)\bigr)
    =
    \left(\frac{u}{1+u}P,1,0\right).
\end{equation*}
These agents receive a lower treatment probability than the other groups. They are not recommended when untreated, and recommended when treated.

The purpose of this design is to save treatment capacity. The target group consists of relatively low types, so the sender would rather reserve treatment for higher types. However, if the target group simply received a lower treatment probability, these agents would prefer to mimic types who receive treatment with probability \(P\).

The mechanism prevents this deviation by attaching recommendation to treatment. Conditional on receiving treatment, a target agent obtains both the treatment payoff \(u\) and the acceptance payoff \(1\). Hence his truthful payoff is
\begin{equation*}
    \frac{u}{1+u}P(1+u)=Pu,
\end{equation*}
which is exactly the payoff from mimicking a type who receives treatment with probability \(P\) but is not recommended. Thus the target group is willing to report truthfully, while the sender reduces the treatment probability spent on them from \(P\) to \(\frac{u}{1+u}P\).

This compensation is not free. Recommending lower-quality agents lowers the average quality of the recommended pool and tightens the receiver's obedience constraint. Hence the sender faces a trade-off. Reducing treatment probability saves scarce capacity, but compensating agents through recommendation makes the obedience constraint harder to satisfy.

This explains the non-monotone structure of the mechanism. The target group is neither the highest group nor the lowest group. It is not optimal to reduce treatment probability for the highest types because the sender wants these agents to receive treatment: they are most valuable for the receiver's obedience constraint. It is also not optimal to target the lowest types because compensating them through recommendation is too costly. Their quality is too low, so recommending them would sharply reduce the posterior mean quality of recommended agents. For the lowest types, the mechanism instead gives the high treatment probability \(P\) but no recommendation. The target group is therefore an intermediate group.

The preceding discussion shows why treatment allocation and recommendation must be designed jointly. The key limitation of the separated design is that it uses the sender's commitment power only toward the receiver, not toward agents. Joint design also uses the sender's commitment power toward agents. By committing ex ante to treatment-contingent recommendation rules, the sender can use future recommendation to screen agents before treatment is allocated. This allows the sender to redirect treatment capacity toward more valuable types. \Cref{cor:joint-improves-separated} shows that this joint design strictly improves on the separated-design benchmark.

\begin{corollary}
\label{cor:joint-improves-separated}
In the nondegenerate interior case, the optimal joint mechanism characterized in \Cref{thm:optimal-direct} achieves a strictly larger mass of accepted agents than the separated-design benchmark in \eqref{eq:separated-design}.
\end{corollary}

\Cref{cor:joint-improves-separated} follows directly from the facts that, in \Cref{thm:optimal-direct}, $P$ is restricted to be strictly greater than $\kappa$, and that the separated-design benchmark corresponds to the same mechanism form as in \Cref{thm:optimal-direct}, with $P=\kappa$. We establish in \Cref{prop:indirect-not-optimal} in \Cref{app:prf-direct} that a mechanism with $P=\kappa$ cannot be optimal in nondegenerate interior cases.

\subsection{Implementation}
\label{subsec:implementation}

The optimal direct mechanism asks agents to report their types. In practice, such direct reports may be unrealistic. We now describe an indirect mechanism that weakly implements the same allocation. Agents do not report their types directly. Instead, they choose between two contracts offered by the sender: a default contract and a non-default contract. The mechanism is described below.

\begin{mechanismbox}{Indirect implementation}

The sender publicly commits to the following mechanism. The sender offers each agent a menu with two contracts: a default contract and a non-default contract. Each agent chooses one contract, and the choice is not observed by the receiver.
Under the non-default contract, the agent receives treatment with fixed probability \(\frac{u}{1+u}P\). The recommendation rule for agents who choose the non-default contract is
\begin{equation*}
    \bigl(q^{N}(\theta,1),q^{N}(\theta,0)\bigr)
    =
    \begin{cases}
    (1,1), & \theta\in(\theta_3,1],\\[0.4em]
    (1,0), & \theta\in(\theta_1,\theta_3],\\[0.4em]
    (0,0), & \theta\in[0,\theta_1].
    \end{cases}
\end{equation*}
Thus, for agents who choose the non-default contract, the recommendation threshold is \(\theta_1\) when treated and \(\theta_3\) when untreated.

Under the default contract, the agent receives treatment with fixed probability \(P\). Agents who choose the default contract face the regular recommendation rule:
\begin{equation*}
    \bigl(q^{D}(\theta,1),q^{D}(\theta,0)\bigr)
    =
    \begin{cases}
    (1,1), & \theta\in(\theta_3,1],\\[0.4em]
    (1,0), & \theta\in(\theta_2,\theta_3],\\[0.4em]
    (0,0), & \theta\in[0,\theta_2].
    \end{cases}
\end{equation*}
Thus, agents who choose the default contract face a higher recommendation threshold when treated: \(\theta_2>\theta_1\). The threshold when untreated remains \(\theta_3\).
\end{mechanismbox}

We show that this menu implements the optimal direct mechanism under tie-breaking in favor of the non-default contract. In the candidate equilibrium, agents in the target group, $\theta\in(\theta_1,\theta_2]$, choose the non-default contract, while all other agents choose the default contract. We verify that no type has a profitable deviation:

 Top-group agents, \(\theta\in(\theta_3,1]\), choose the default contract. Under the default contract, they receive treatment with probability \(P\) and are recommended regardless of treatment status. Under the non-default contract, they are still recommended regardless of treatment status, but receive treatment with the lower probability \(\frac{u}{1+u}P\). Hence they strictly prefer the default contract.

Upper-middle agents, \(\theta\in(\theta_2,\theta_3]\), also choose the default contract. Under the default contract, they receive treatment with probability \(P\) and are recommended whenever treated. Under the non-default contract, they are also recommended whenever treated, but receive treatment with the lower probability \(\frac{u}{1+u}P\). Hence they strictly prefer the default contract.

Bottom agents, \(\theta\in[0,\theta_1]\), choose the default contract as well. They are never recommended under either contract. Choosing the non-default contract only lowers their treatment probability, so they strictly prefer the default contract.

Finally, consider target agents, \(\theta\in(\theta_1,\theta_2]\). If a target agent chooses the default contract, he receives treatment with probability \(P\) but is never recommended, so his payoff is \(Pu\). If he chooses the non-default contract, he receives treatment with probability \(\frac{u}{1+u}P\), and treatment is bundled with recommendation. His payoff is
\begin{equation*}
    \frac{u}{1+u}P(1+u)=Pu.
\end{equation*}
Thus target agents are indifferent between the two contracts. With tie-breaking in favor of the non-default contract, all target agents choose it.

The induced allocation coincides with the optimal direct mechanism. Target agents choose the non-default contract, receive treatment with probability \(\frac{u}{1+u}P\), and are recommended when treated. All other agents choose the default contract, receive treatment with probability \(P\), and face the regular recommendation rule. Hence top agents are recommended regardless of treatment status, upper-middle agents are recommended only when treated, and bottom agents are never recommended.

This menu induces self-selection. Choosing the non-default contract means giving up treatment probability in exchange for a lower recommendation threshold when treated. Only the target group finds this trade-off weakly attractive, so the sender identifies them through their contract choice.

\subsection{A Proof Sketch of \topref{Theorem }{thm:optimal-direct}}

In this subsection, we present a proof sketch for our main result, \Cref{thm:optimal-direct}. A complete proof is provided in \Cref{app:prf-direct}.

The proof starts with the following two observations. First, the probability of recommendation is the only instrument that the sender can use to punish agents who are found to have misreported their types. Since the sender has full commitment power, it is optimal to commit to the harshest punishment for misreports by never recommending agents whose reported type is inconsistent with their true type, i.e., $q(\htt,\theta,d)=0$ for all $\htt\neq\theta$, $d\in\{0,1\}$. As a result, the only potential benefit for an agent from deviating from truthful reporting is a higher probability of treatment. The treatment allocation rule $p$ depends only on the reported type. Therefore, if an agent is to deviate, it is optimal for him to report a type that receives the highest probability of treatment. We denote this highest probability by $P$, where $P=\sup_{\theta\in\Theta}p(\theta)$.

The second observation is that, apart from treatment, agents care only about the ex-ante probability of being recommended. That is, conditional on this probability, it does not matter whether recommendation occurs when they are treated or untreated. We denote the ex-ante recommendation rule by $Q$, where $Q(\theta)$ is the ex-ante probability of a type-$\theta$ agent being recommended given truthful reporting, i.e., $Q(\theta)=p(\theta)q(\theta,\theta,1)+[1-p(\theta)]q(\theta,\theta,0)$. It is sufficient for the sender to design $(p,Q)$ as the incentive scheme for the agents. Moreover, the same agent has higher quality when treated than when untreated. Therefore, for any fixed ex-ante recommendation probability $Q(\theta)$, the sender prefers to assign recommendation probability to the treated state first. Only after the agent is recommended with probability one when treated does the sender assign positive recommendation probability to the untreated state.


With these two observations, it can be shown that the sender's problem \eqref{eq:sender-primal} is equivalent to the following problem:
\begin{equation*}
\begin{aligned}
    \max_{p,Q,P}\quad
    &\int_{\Theta}
        Q(\theta)\D F(\theta),
     \\
    \text{s.t.}\quad
    & \int_{\Theta} p(\theta)\,\D F(\theta)\leq \kappa,
    && \text{(Capacity)} \\
    & \int_\Theta [(\theta-\tau) Q(\theta)+v \min\{p(\theta),Q(\theta)\}]\D F(\theta)\geq 0,
    && \text{(Obedience)} \\
    & Q(\theta)+p(\theta) u\geq Pu,
    \quad \forall \theta\in\Theta,
    && \text{(IC$_1$)}\\
    & p(\theta)\leq P,
    \quad \forall \theta\in\Theta,
    && \text{(IC$_2$)}
\end{aligned}
\end{equation*}
The optimal recommendation rule $q$ can be recovered from $p$ and $Q$.

To solve this problem, we first fix the value of $P$ and optimize over $p$ and $Q$. When $P$ is fixed, the problem is a convex optimization problem satisfying Slater's condition. We therefore proceed with Lagrangian methods and show that given $P$, the solution takes either of the following two forms:
\begin{equation*}
    (p(\theta),Q(\theta))=\switch{
        &(P,1), &\theta_3<\theta&\leq1\\
        &(P,P), &\theta_2<\theta&\leq\theta_3\\
        &\textstyle\left(\frac u{1+u}P,\frac u{1+u}P\right), &\theta_1<\theta&\leq\theta_2\\
        &(P,0), &0\leq\theta&\leq\theta_1
    }\text{ or }\switch{
        &(0,1), &\theta_3<\theta&\leq1\\
        &(0,uP), &\tilde\theta_2<\theta&\leq\theta_3\\
        &\textstyle\left(\frac u{1+u}P,\frac u{1+u}P\right), &\theta_1<\theta&\leq\tilde\theta_2\\
        &(P,0), &0\leq\theta&\leq\theta_1
    },
\end{equation*}
where $\theta_1$, $\theta_2$, $\tilde\theta_2$, and $\theta_3$ are cutoffs, whose values are pinned down by the Lagrangian multipliers.

We then optimize over $P\in[\kappa,1]$. We assume that at the optimum, the sender can recommend a strictly positive measure of agents, otherwise the sender’s maximal payoff is zero. It can be shown that for a given $P$, if the solution takes the second form, it is strictly dominated by the solution when $P=\kappa$. Furthermore, it can be shown that $P=\kappa$ is also not optimal. Therefore, the optimal direct mechanism takes the first form with $P>\kappa$, which represents a mechanism that takes exactly the four-segment form presented in \Cref{thm:optimal-direct}. 

\section{Applications}
\label{sec:applications}

The model has a rich set of applications. This section discusses three examples: education and placement, industrial policy, and startup incubation. 

\subsection{Education and Placement}
A natural application is education. The sender is a school or program; the agents are students; and the receiver is the labor market, employers, or another institution that evaluates students after they leave the school. The treatment is a scarce educational resource, such as access to strong instructors, small honors sections, research positions, mentoring, internships, or placement support.

When these resources are allocated, the school does not yet know students' true ability. Admission files may not reveal which students will become strong candidates later, while students may have private information about their preparation, motivation, or fit. All students value scarce educational resources and successful placement.

After students take courses, complete projects, and interact with faculty, the school learns much more about their ability. It can then recommend students through grades, honors, or recommendation letters. Employers do not observe all of the school's internal information, so they decide whether to hire recommended students based on the expected quality of the recommended pool.

This maps directly into the model. The school wants to maximize the number of students accepted by employers. It has a limited treatment that improves student quality, but must allocate it before fully learning students' ability. Later, once ability is better observed, the school can recommend students to employers. The receiver's obedience constraint captures the credibility of these recommendations: if the school recommends too many weak students, employers discount the signal.

The model suggests that educational resources and placement support should be designed jointly. A school can implement this logic through honors programs or academic tracks. These tracks may differ both in access to scarce educational resources and in the probability of later recommendation. Students in a less resource-intensive track may receive fewer educational inputs, but may also face stronger placement support or a lower threshold for endorsement if they perform well.

Such tracks induce self-selection. Students choose based not only on current resources, but also on the future recommendation attached to each track. The school can therefore use track choice to identify a target group before it fully observes students' ability, implementing the same logic as the optimal joint mechanism.

\subsection{Industrial Policy and Firm Certification}
A second application is local industrial policy. The sender is a local government; the agents are firms; and the receiver is an upper-level government. The treatment is limited policy support, such as access to an industrial park, subsidized land, tax relief, talent-policy support, loan guarantees, infrastructure support, or procurement opportunities.

The local government wants to cultivate firms that can later receive recognition from an upper-level government as promising firms in strategically important sectors. Such recognition may serve as evidence of local policy success.\footnote{For example, China's ``specialized and innovative'' SME policies emphasize local support for firm development while also linking local implementation to upper-level recognition. A 2021 Ministry of Finance and Ministry of Industry and Information Technology notice asks local governments to improve support policies and public-service systems, with the goal of helping about 10,000 SMEs grow into national-level ``Little Giant'' firms. A 2024 notice further emphasizes central-local coordination, local cultivation services, provincial recommendation of candidate firms, and performance evaluation. See Ministry of Finance and Ministry of Industry and Information Technology, ``Notice on Supporting the High-Quality Development of Specialized and Innovative SMEs,'' 2021, \url{https://www.beijing.gov.cn/zhengce/zhengcefagui/qtwj/202204/t20220413_2675299.html}; and ``Notice on Further Supporting the High-Quality Development of Specialized and Innovative SMEs,'' 2024, \url{https://jjs.mof.gov.cn/zhengcefagui/202406/t20240618_3937446.htm}.}
 When local support is allocated, however, the government cannot perfectly distinguish high-potential firms from low-potential firms mainly seeking subsidies. Firms know more about their own productivity, technology, commitment, and growth prospects. All firms value policy support and recognition by the upper-level government.

Over time, the local government learns much more about firm quality. It observes production, employment, tax payments, project completion, orders, financing, patents, exports, and other performance measures. It can then recommend firms to the upper-level government. The upper-level government does not observe all local information, so it relies on the local government's recommendation. The recommendation must remain credible: if too many weak firms are recommended, the upper-level government will discount the signal.

This setting captures a central difficulty in industrial policy. Early policy support is valuable, but it may attract weak firms mainly interested in obtaining subsidies. The local government would like scarce support to reach firms that can later become credible candidates for upper-level recognition, but these firms are difficult to identify at the initial allocation stage.

The model highlights that industrial policy is not only about allocating subsidies. It is also about designing later certification. A local government may use future recommendation to induce firms to reveal information before scarce support is allocated. Some firms may accept less policy support in exchange for a better chance of later certification if they meet performance standards. This allows the government to save scarce resources while still identifying firms that can become valuable candidates for upper-level recognition.
\subsection{Incubators and Accelerators}

A third application is startup incubation. The sender is an incubator, accelerator, university entrepreneurship center, or local startup platform; the agents are startups; and the receiver is a group of later-stage investors, strategic partners, acquirers, or other outside evaluators. The treatment is a scarce developmental input, such as seed funding, office space, technical assistance, mentorship, legal support, cloud credits, network access, or investor introductions.

Early in the process, the incubator does not know which startups are truly high quality. Many startups are still at the idea or prototype stage, and the available information may consist mainly of founders' claims, pitch decks, and limited early traction. Founders may know more about their own technology, effort, or market fit than the incubator does. All startups value support from the incubator, and all startups value being recommended to investors or partners later.

After the incubation period, the incubator learns much more. It observes whether the team builds the product, whether customers respond, whether founders execute, and whether the business model improves. At that point, the incubator can recommend startups through demo days, curated introductions, endorsement letters, or access to a preferred investor network. Investors do not observe all of the incubator's private information, but they may trust the recommendation if the recommended pool is sufficiently strong.

This application captures the joint role of early support and later certification. If early support is allocated without using the later recommendation rule, some resources may go to startups that never become investor-ready. The incubator would instead like support to reach startups that are more likely to become credible candidates for later financing, but these startups are difficult to identify at the beginning.

The model suggests that incubators should design early support and later investor access jointly. Some startups may accept less early support in exchange for a better chance of later investor introduction if they reach certain milestones. This contract is attractive only to startups that expect to benefit from the later recommendation. The incubator can therefore use the promise of future certification to identify a target group and reserve more intensive resources for startups that are more valuable to the later recommendation pool.\\

\section{Conclusion}
\label{sec:conclusion}
This paper studies a joint mechanism and information design problem. A sender wants to maximize the mass of agents accepted by a receiver. The sender has two instruments: a scarce treatment that can improve agents' quality, and a later recommendation that affects the receiver's belief. The key friction is informational timing. Treatment must be allocated before the sender fully observes agents' types, while recommendation occurs later, after more information is revealed.

The main result is that treatment allocation and recommendation should be designed jointly. The optimal joint mechanism targets an intermediate group: the sender reduces treatment probability not for the lowest or highest types, but for agents in the middle. These agents receive a lower probability of treatment, but are compensated by a recommendation when treatment is realized. This non-monotone treatment allocation allows the sender to save treatment capacity and redirect it toward higher types, whose treated quality is more valuable for satisfying the receiver's acceptance constraint. The mechanism therefore uses recommendation in two ways. It communicates information to the receiver, and it also helps elicit information from agents before treatment is allocated.

The model applies to settings such as education, local industrial policy, and startup incubation, where an institution allocates scarce resources before fully learning agents' quality and later recommends selected agents to an outside evaluator. Across these applications, the central takeaway is the same: later recommendation should not be treated only as an information device, but also as an incentive device that improves the early allocation of scarce resources.

The model also speaks to empirical evaluations of programs. Lack of direct type observability is sometimes used as part of the argument for random or quasi-random program assignment. Our model suggests caution with this interpretation. A program designer who seeks to maximize program success may instead use a non-random treatment-allocation rule based on agents' messages, applications, or self-selection decisions. In this case, treatment assignment can be systematically related to agents' private information even though the designer does not observe types directly. Lack of direct type observability is therefore not sufficient for random assignment.

\clearpage
\begin{singlespace}
\bibliographystyle{aer}
\bibliography{our-cites.bib}

@article{KamenicaGentzkow2011,
  author  = {Kamenica, Emir and Gentzkow, Matthew},
  title   = {Bayesian Persuasion},
  journal = {American Economic Review},
  year    = {2011},
  volume  = {101},
  number  = {6},
  pages   = {2590--2615},
  doi     = {10.1257/aer.101.6.2590}
}

@article{RayoSegal2010,
  author  = {Rayo, Luis and Segal, Ilya},
  title   = {Optimal Information Disclosure},
  journal = {Journal of Political Economy},
  year    = {2010},
  volume  = {118},
  number  = {5},
  pages   = {949--987},
  doi     = {10.1086/657922}
}

@article{BergemannMorris2016,
  author  = {Bergemann, Dirk and Morris, Stephen},
  title   = {Information Design, {Bayesian} Persuasion, and {Bayes} Correlated Equilibrium},
  journal = {American Economic Review},
  year    = {2016},
  volume  = {106},
  number  = {5},
  pages   = {586--591},
  doi     = {10.1257/aer.p20161046}
}

@article{BergemannMorris2019,
  author  = {Bergemann, Dirk and Morris, Stephen},
  title   = {Information Design: A Unified Perspective},
  journal = {Journal of Economic Literature},
  year    = {2019},
  volume  = {57},
  number  = {1},
  pages   = {44--95},
  doi     = {10.1257/jel.20181489}
}

@article{Kolotilin2018,
  author  = {Kolotilin, Anton},
  title   = {Optimal Information Disclosure: A Linear Programming Approach},
  journal = {Theoretical Economics},
  year    = {2018},
  volume  = {13},
  number  = {2},
  pages   = {607--635},
  doi     = {10.3982/TE1805}
}

@article{Kamenica2019,
  author  = {Kamenica, Emir},
  title   = {{Bayesian} Persuasion and Information Design},
  journal = {Annual Review of Economics},
  year    = {2019},
  volume  = {11},
  pages   = {249--272},
  doi     = {10.1146/annurev-economics-080218-025739}
}

@article{MussaRosen1978,
  author  = {Mussa, Michael and Rosen, Sherwin},
  title   = {Monopoly and Product Quality},
  journal = {Journal of Economic Theory},
  year    = {1978},
  volume  = {18},
  number  = {2},
  pages   = {301--317},
  doi     = {10.1016/0022-0531(78)90085-6}
}

@article{Myerson1981,
  author  = {Myerson, Roger B.},
  title   = {Optimal Auction Design},
  journal = {Mathematics of Operations Research},
  year    = {1981},
  volume  = {6},
  number  = {1},
  pages   = {58--73},
  doi     = {10.1287/moor.6.1.58}
}

@article{BaronMyerson1982,
  author  = {Baron, David P. and Myerson, Roger B.},
  title   = {Regulating a Monopolist with Unknown Costs},
  journal = {Econometrica},
  year    = {1982},
  volume  = {50},
  number  = {4},
  pages   = {911--930},
  doi     = {10.2307/1912769}
}

@article{MaskinRiley1984,
  author  = {Maskin, Eric and Riley, John},
  title   = {Monopoly with Incomplete Information},
  journal = {The RAND Journal of Economics},
  year    = {1984},
  volume  = {15},
  number  = {2},
  pages   = {171--196},
  doi     = {10.2307/2555674}
}

@book{LaffontMartimort2002,
  author    = {Laffont, Jean-Jacques and Martimort, David},
  title     = {The Theory of Incentives: The Principal-Agent Model},
  publisher = {Princeton University Press},
  address   = {Princeton, NJ},
  year      = {2002}
}

@article{BergemannPesendorfer2007,
  author  = {Bergemann, Dirk and Pesendorfer, Martin},
  title   = {Information Structures in Optimal Auctions},
  journal = {Journal of Economic Theory},
  year    = {2007},
  volume  = {137},
  number  = {1},
  pages   = {580--609},
  doi     = {10.1016/j.jet.2007.02.001}
}

@article{BergemannHeumannMorris2026a,
  author  = {Bergemann, Dirk and Heumann, Tibor and Morris, Stephen},
  title   = {Screening with Persuasion},
  journal = {Journal of Political Economy},
  year    = {2026},
  doi     = {10.1086/738342},
  note    = {Forthcoming}
}

@article{Dworczak2020,
  author  = {Dworczak, Piotr},
  title   = {Mechanism Design with Aftermarkets: Cutoff Mechanisms},
  journal = {Econometrica},
  year    = {2020},
  volume  = {88},
  number  = {6},
  pages   = {2629--2661},
  doi     = {10.3982/ECTA15768}
}

@article{Lizzeri1999,
  author  = {Lizzeri, Alessandro},
  title   = {Information Revelation and Certification Intermediaries},
  journal = {The RAND Journal of Economics},
  year    = {1999},
  volume  = {30},
  number  = {2},
  pages   = {214--231},
  doi     = {10.2307/2556078}
}

@article{DranoveJin2010,
  author  = {Dranove, David and Jin, Ginger Zhe},
  title   = {Quality Disclosure and Certification: Theory and Practice},
  journal = {Journal of Economic Literature},
  year    = {2010},
  volume  = {48},
  number  = {4},
  pages   = {935--963},
  doi     = {10.1257/jel.48.4.935}
}

@article{Zapechelnyuk2020,
  author  = {Zapechelnyuk, Andriy},
  title   = {Optimal Quality Certification},
  journal = {American Economic Review: Insights},
  year    = {2020},
  volume  = {2},
  number  = {2},
  pages   = {161--176},
  doi     = {10.1257/aeri.20190387}
}

@misc{BoleslavskyKim2018,
  author = {Boleslavsky, Raphael and Kim, Kyungmin},
  title  = {{Bayesian} Persuasion and Moral Hazard},
  year   = {2018},
  note   = {SSRN working paper},
  doi    = {10.2139/ssrn.2913669}
}

@misc{SaeediShourideh2026,
  author = {Saeedi, Maryam and Shourideh, Ali},
  title  = {Optimal Rating Design under Moral Hazard},
  year   = {2026},
  note   = {Working paper, arXiv:2008.09529}
}

@misc{Xiao2024,
  author = {Xiao, Peiran},
  title  = {Incentivizing Agents through Ratings},
  year   = {2024},
  note   = {Working paper, arXiv:2407.10525}
}

@inproceedings{CamboniCarnehl2025,
  author    = {Camboni, Matteo and Carnehl, Christoph},
  title     = {Inputs or Outputs: What to Test and How to Test},
  booktitle = {Proceedings of the 26th ACM Conference on Economics and Computation},
  year      = {2025},
  pages     = {577},
  doi       = {10.1145/3736252.3742597}
}

@article{MartimortSemenov2006,
  author  = {Martimort, David and Semenov, Aggey},
  title   = {Continuity in Mechanism Design without Transfers},
  journal = {Economics Letters},
  year    = {2006},
  volume  = {93},
  number  = {2},
  pages   = {182--189},
  doi     = {10.1016/j.econlet.2006.04.014}
}

@incollection{SchummerVohra2007,
  author    = {Schummer, James and Vohra, Rakesh V.},
  title     = {Mechanism Design without Money},
  booktitle = {Algorithmic Game Theory},
  editor    = {Nisan, Noam and Roughgarden, Tim and Tardos, {\'E}va and Vazirani, Vijay V.},
  publisher = {Cambridge University Press},
  address   = {Cambridge},
  year      = {2007},
  chapter   = {10},
  pages     = {243--265}
}

@article{Miralles2012,
  author  = {Miralles, Antonio},
  title   = {Cardinal {Bayesian} Allocation Mechanisms without Transfers},
  journal = {Journal of Economic Theory},
  year    = {2012},
  volume  = {147},
  number  = {1},
  pages   = {179--206},
  doi     = {10.1016/j.jet.2011.11.002}
}

@article{kolotilin2017persuasion,
  title={Persuasion of a privately informed receiver},
  author={Kolotilin, Anton and Mylovanov, Tymofiy and Zapechelnyuk, Andriy and Li, Ming},
  journal={Econometrica},
  volume={85},
  number={6},
  pages={1949--1964},
  year={2017},
  publisher={Wiley Online Library}
}

@article{ben2014,
  title={Optimal allocation with costly verification},
  author={Ben-Porath, Elchanan and Dekel, Eddie and Lipman, Barton L},
  journal={American Economic Review},
  volume={104},
  number={12},
  pages={3779--3813},
  year={2014},
  publisher={American Economic Association 2014 Broadway, Suite 305, Nashville, TN 37203}
}

@article{zubrickas2015optimal,
  title={Optimal grading},
  author={Zubrickas, Robertas},
  journal={International Economic Review},
  volume={56},
  number={3},
  pages={751--776},
  year={2015},
  publisher={Wiley Online Library}
}

@article{chopra2025incentive,
  title={Incentive-Compatible Information Design},
  author={Chopra, Hershdeep and Ely, Jeffrey},
  year={2025}
}
\end{singlespace}


\clearpage

\begin{appendices}

\crefalias{section}{appendix}
\crefalias{subsection}{appendix}
\renewcommand{\thesubsection}{\Alph{section}.\arabic{subsection}} 
\renewcommand{\theequation}{\Alph{section}.\arabic{equation}}
\counterwithin{lemma}{section}
\renewcommand{\thelemma}{\Alph{section}\arabic{lemma}}
\counterwithin{proposition}{section}
\renewcommand{\theproposition}{\Alph{section}\arabic{proposition}}
\counterwithin{figure}{section}
\counterwithin{table}{section}

\section{Proof of \topref{Theorem }{thm:optimal-direct}}
\label{app:prf-direct}

In this section, we provide a characterization for the solution to problem \eqref{eq:sender-primal}, which is restated below.
\begin{equation}
\begin{aligned}
    \max_{p,q}\quad
    &
    \int_{\Theta}
    \left[
        p(\theta)q(\theta,\theta,1)
        +
        (1-p(\theta))q(\theta,\theta,0)
    \right]\D F(\theta), \\
    \text{s.t.}\quad
    & \int_{\Theta} p(\theta)\,\D F(\theta)\leq \kappa,
    && \text{(Capacity)} \\
    & \mu_\pi\geq \tau,
    && \text{(Obedience)} \\
    & U_\pi(\theta,\theta)\geq U_\pi(\theta,\hat{\theta}),
    \quad \forall \theta,\hat{\theta}\in\Theta,
    && \text{(IC)}
\end{aligned}
\tag{\ref*{eq:sender-primal}}
\end{equation}
.



\begin{lemma}
    \label{lem:punish}
    There exists an optimal mechanism $(p,q)$ that satisfies $q(\hat\theta,\theta,d)=0$ for all $\hat\theta\neq \theta$ and $d\in\{0,1\}$.
\end{lemma}

\begin{proof}
    In problem \eqref{eq:sender-primal}, when $\hat\theta\neq \theta$, $q(\hat\theta,\theta,d)$ shows up only in the IC constraint, which expands to 
    \begin{equation*}
    p(\theta)
    [
        u+q(\theta,\theta,1)
    ]
    +
    (1-p(\theta))
    q(\theta,\theta,0)
    \geq 
    p(\hat{\theta})
    [
        u+q(\hat{\theta},\theta,1)
    ]
    +
    (1-p(\hat{\theta}))
    q(\hat{\theta},\theta,0).
    \end{equation*}
    Clearly, the right-hand side is increasing in $q(\htt,\theta,0)$ and $q(\htt,\theta,1)$. Therefore, it is optimal to set them such that the IC constraint is as slack as possible, i.e., $q(\htt,\theta,0)=q(\htt,\theta,1)=0$ whenever $\htt\neq \theta$. This is essentially assigning the harshest punishment available whenever a lie is detected.
\end{proof}

\Cref{lem:punish} implies that we can restrict our attention to the mechanisms that impose the harshest punishment to agents that misreport.

We define the ```ex-ante recommendation rule'' $Q:[0,1]\to[0,1]$ as
\begin{equation*}
    Q(\theta):=p(\theta)q(\theta,\theta,1)+[1-p(\theta)]q(\theta,\theta,0).
\end{equation*}
``Ex-ante'' here means before assigning the treatment.

With \Cref{lem:punish}, problem \eqref{eq:sender-primal} can be rewritten using the definition of $Q$ as
\begin{equation}
\label{eq:sender-primal-w-Q}
\begin{aligned}
    \max_{p,q}\quad
    &\int_{\Theta}
        Q(\theta)\D F(\theta),
     \\
    \text{s.t.}\quad
    & \int_{\Theta} p(\theta)\,\D F(\theta)\leq \kappa,
    && \text{(Capacity)} \\
    & \int_\Theta [(\theta-\tau) Q(\theta)+v p(\theta)q(\theta,\theta,1)]\D F(\theta)\geq 0,
    && \text{(Obedience)} \\
    & Q(\theta)+p(\theta) u\geq p(\hat\theta) u,
    \quad \forall \theta,\hat{\theta}\in\Theta.
    && \text{(IC)}
\end{aligned}
\end{equation}
The obedience constraint is expanded and simplified according to the definition of $\mu_\pi$. The IC constraint can be further rewritten as
\begin{align*}
    Q(\theta)+p(\theta)u\geq \sup_{\hat\theta\in\Theta} p(\hat\theta)u,\quad\forall \theta\in\Theta.
\end{align*}

\begin{lemma}
\label{lem:A-is-minpq}
    The optimal mechanism $(p,q)$ satisfies
    \begin{align*}
        q(\theta,\theta,1)&=\min\left\{1,\frac{Q(\theta)}{p(\theta)}\right\}, \\
        q(\theta,\theta,0)&=\max\left\{\frac{Q(\theta)-p(\theta)}{1-p(\theta)},0\right\}.
    \end{align*}
    When $p(\theta)\in\{0,1\}$, the value of $q(\theta,\theta,1-p(\theta))$ is irrelevant.
\end{lemma}

\begin{proof}
    Given fixed $p(\theta)$ and $Q(\theta)$, the obedience constraint becomes slacker when $q(\theta,\theta,1)$ increases. Therefore, it is optimal to set $q(\theta,\theta,1)$ to be its maximal feasible value. The only two constraints on the upper bound of $q(\theta,\theta,1)$ are $q(\theta,\theta,1)\leq 1$, and
    \begin{equation*}
        p(\theta)q(\theta,\theta,1)+[1-p(\theta)]q(\theta,\theta,0)=Q(\theta)\leq 1\implies q(\theta,\theta,1)\leq \frac{Q(\theta)}{p(\theta)},
    \end{equation*}
    because $p(\theta)$, $q(\theta,\theta,0)$ and $1-p(\theta)$ are non-negative. Hence, the optimal $q(\theta,\theta,1)$ is $\min\{1,\frac{Q(\theta)}{p(\theta)}\}$, and the corresponding optimal $q(\theta,\theta,0)$ takes the value such that $p(\theta)q(\theta,\theta,1)+[1-p(\theta)]q(\theta,\theta,0)=Q(\theta)$.
\end{proof}

\Cref{lem:A-is-minpq} implies that it is sufficient for the sender to design $p$ and ex-ante recommendation rule $Q$ as the mechanism. The optimal treatment-specific recommendation rule $q(\theta,\theta,0)$ and $q(\theta,\theta,1)$ can be recovered from the optimal $p$ and $Q$. This further simplifies the sender's problem into 
\begin{equation}
\label{eq:sender-prob}
\begin{aligned}
    \max_{p,Q}\quad
    &\int_{\Theta}
        Q(\theta)\D F(\theta),
     \\
    \text{s.t.}\quad
    & \int_{\Theta} p(\theta)\,\D F(\theta)\leq \kappa,
    && \text{(Capacity)} \\
    & \int_\Theta [(\theta-\tau) Q(\theta)+v \min\{p(\theta),Q(\theta)\}]\D F(\theta)\geq 0,
    && \text{(Obedience)} \\
    & Q(\theta)+p(\theta) u\geq \sup_{\htt\in\Theta}p(\hat\theta) u,
    \quad \forall \theta\in\Theta.
    && \text{(IC)}
\end{aligned}
\end{equation}
We hereby refer to $(p,Q)$ also as a mechanism.

\begin{lemma}
\label{lem:cap-bind}
    There exists an optimal solution to problem \eqref{eq:sender-prob} such that the capacity constraint binds.
\end{lemma}
\begin{proof}
    Suppose $(p,Q)$ is an optimal solution to \eqref{eq:sender-prob} such that
    \begin{equation*}
        k:=\int_\Theta p(\theta)\D F(\theta) < \kappa.
    \end{equation*}
    We define 
    \begin{equation*}
        \tilde p(\theta):=p(\theta)+\frac{\kappa-k}{1-k}(1-p(\theta)).
    \end{equation*}
    It is easy to verify that $\tilde p:[0,1]\to [0,1]$, 
    \begin{equation*}
        \int_\Theta \tilde p(\theta)\D F(\theta) = \kappa,
    \end{equation*}
    and obedience and IC constraints still hold for $(\tilde p,Q)$. Therefore, $(\tilde p,Q)$ is an optimal solution to \eqref{eq:sender-prob} such that the capacity constraint binds.
\end{proof}

\Cref{lem:cap-bind} implies that $\sup_{\htt\in\Theta}p(\htt)\geq \kappa$ at the optimum. Define $P=\sup_{\htt\in\Theta}p(\htt)$. By \Cref{lem:cap-bind}, we restrict our attention to the case where $P\in[\kappa,1]$. To solve \eqref{eq:sender-prob}, we follow a two-step procedure to linearize the constraints. We first characterize the optimal $(p,Q)$ for a given $P$ by solving
\begin{equation}
\label{eq:fix-P-max-over-p-Q}
\begin{aligned}
    \max_{p,Q}\quad
    &\int_{\Theta}
        Q(\theta)\D F(\theta),
     \\
    \text{s.t.}\quad
    & \int_{\Theta} p(\theta)\,\D F(\theta)\leq \kappa,
    && \text{(Capacity)} \\
    & \int_\Theta [(\theta-\tau) Q(\theta)+v \min\{p(\theta),Q(\theta)\}]\D F(\theta)\geq 0,
    && \text{(Obedience)} \\
    & Q(\theta)+p(\theta) u\geq Pu,
    \quad \forall \theta\in\Theta,
    && \text{(IC$_1$)}\\
    & p(\theta)\leq P,
    \quad \forall \theta\in\Theta,
    && \text{(IC$_2$)}
\end{aligned}
\end{equation}

We then substitute the resulting solution $(p_P,Q_P)$ into the objective and choose $P$ to solve
\begin{equation*}
\label{eq:max-over-P}
    \max_{P\in[\kappa,1]}\ \int_{\Theta}
        Q_P(\theta)\D F(\theta),\quad\st\ (p_P,Q_P)\text{ solves \eqref{eq:fix-P-max-over-p-Q}}.
\end{equation*}

To solve problem \eqref{eq:fix-P-max-over-p-Q}, we introduce Lagrangian multipliers for the capacity and obedience constraints:
\begin{equation}
    \label{eq:lagrangian}
\begin{aligned}
    \max_{p,Q}\ &\int_\Theta \bigl[Q(\theta)-\lambda p(\theta)+\eta((\theta-\tau) Q(\theta) + v \min\{p(\theta),Q(\theta)\})\bigr]\D F(\theta)+\lambda\kappa,\\
    \st&\text{(IC$_1$) } Q(\theta)+p(\theta)u\geq Pu,\quad\forall \theta\in\Theta,\\
    &\text{(IC$_2$) } p(\theta)\leq P,\quad\forall \theta\in\Theta.\\
\end{aligned}
\end{equation}
    
\begin{lemma}
    \label{lem:lagrangian}
    For a fixed $P\in[\kappa,1]$, if the constraint set of \eqref{eq:fix-P-max-over-p-Q} is non-empty, $(p^*,Q^*)$ is the solution to problem \eqref{eq:fix-P-max-over-p-Q} if and only if there exists $\lambda,\eta\geq0$ satisfying the KKT conditions such that $(p^*,Q^*)$ is also a solution to problem \eqref{eq:lagrangian}.
\end{lemma}

\begin{proof}
    In problem \eqref{eq:fix-P-max-over-p-Q}, the objective function is linear in $(p,Q)$. All constraints except the obedience constraint are affine in $(p,Q)$. It is easy to verify that $\min\{p(\theta),Q(\theta)\}$ is a concave function in $(p(\theta),Q(\theta))$, and therefore the left-hand side of the obedience constraint is concave in $(p,Q)$. Hence, the constraint set is a convex set, and therefore problem \eqref{eq:fix-P-max-over-p-Q} is a convex optimization problem.

    We can linearize the constraint set by replacing $\min\{p(\theta),Q(\theta)\}$ with $m(\theta)$ and introducing constraint $p(\theta),Q(\theta)\geq m(\theta)\geq 0$. The new constraint set is affine in $(p,Q,m)$. For an affine constraint set, the Slater's condition only requires that a feasible point exists. As long as the original constraint set is non-empty, $(p,Q,\min\{p,Q\})$ is feasible in the linearized constraint set. Therefore, the Slater's condition is satisfied, which finishes the proof.
\end{proof}

\Cref{lem:lagrangian} implies that every positive-payoff solution relevant for problem \eqref{eq:fix-P-max-over-p-Q} is supported by problem \eqref{eq:lagrangian}. Hence, to characterize the optimal direct mechanism, it is sufficient to characterize the solutions to \eqref{eq:lagrangian} and then find a specific combination of multipliers such that its corresponding solution to \eqref{eq:lagrangian} also satisfies the capacity and obedience constraints. 

Notice that the objective function of \eqref{eq:fix-P-max-over-p-Q} is increasing in $Q$, and that the IC constraint becomes slacker when $Q$ increases. Since the obedience constraint is continuous in $Q$, as long as it is slack, the objective can be improved by increasing $Q$ slightly without violating the constraints. Therefore, at optimum, the obedience constraint binds, and \Cref{lem:lagrangian} allows the
supporting multipliers to be chosen with $\eta>0$.

Given $\lambda$ and $\eta$, problem \eqref{eq:lagrangian} is a pointwise maximization problem with linear constraints. Let $(p_{\lambda,\eta},Q_{\lambda,\eta})$ be the solution to problem \eqref{eq:lagrangian} given the multipliers $\lambda$ and $\eta$. We further define
\begin{align*}
    \Lambda=\frac\lambda\eta,\ \Eta=\utt-\frac 1\eta.
\end{align*}

\begin{proposition}
    \label{prop:beautiful-graph}\
    Given $\lambda\geq0$ and $\eta>0$, the solution to problem \eqref{eq:lagrangian} is given by
    \begin{equation}
    \label{eq:beautiful-graph-sol}
        (p_{\lambda,\eta}(\theta),Q_{\lambda,\eta}(\theta))=\switch{
         &(P,1), &&\Eta\leq \theta ,\ \Lambda\leq v,\\
         &(0,1), &&\Eta\leq \theta ,\ \Lambda\geq v,\\
         &(P,P), &&\Eta\geq \theta ,\ \Lambda+\Eta\leq \theta +v,\\
         &(P,0),&&\Eta\geq \theta ,\ \Eta-\frac \Lambda u\geq \theta +v,\\
         &(0,uP), &&\Eta\geq \theta ,\ \Eta-\frac\Lambda u\leq \theta -\frac vu,\\
         &\left(\frac{u}{1+u}P,\frac{u}{1+u}P\right), &&\text{otherwise}.
        }
    \end{equation}
\end{proposition}

\begin{proof}
Problem \eqref{eq:lagrangian} is a pointwise optimization problem. For a given $\theta$, we introduce $m=\min\{p(\theta),Q(\theta)\}$ to linearize the problem. Equivalently, we can rewrite the problem into the following linear programming problem with $\Lambda$ and $\Eta$:
\begin{align*}
    \max_{p(\theta),Q(\theta),m}\ &-\Lambda\cdot p(\theta)+(\theta -\Eta)\cdot Q(\theta)+v\cdot m\\
    \st\ &Q(\theta)+p(\theta)u\geq Pu,\\
    & p(\theta)\leq P,\\
    & 0\leq m\leq p(\theta),Q(\theta)\leq 1.
\end{align*}
Here, we scaled the objective function by $\eta^{-1}$. We know that $\Lambda\geq 0$.

The objective function is optimized at at least one of the extreme points of the constraint set. All extreme points such that at least one of $m=p(\theta)$ and $m=Q(\theta)$ holds and the corresponding values of the objective function are listed in \Cref{tab:extreme-and-value}. For convenience, we label the values $V_1$--$V_6$. The solution to the problem is the extreme point with the highest corresponding objective function value. 

\begingroup
\renewcommand*{\arraystretch}{1.2}
\begin{table}[htbp]
\centering
\begin{tabular}{cr}
\toprule
Extreme point $(p(\theta),Q(\theta),m)$                               & \multicolumn{1}{c}{Value of objective}                   \\ \midrule
$(P,1,P)$                                                   & $(v-\Lambda)P+\theta -\Eta=:V_1$              \\
$(0,1,0)$                                                   & $\theta -\Eta=:V_2$                           \\
$(P,P,P)$                                                   & $(v-\Lambda+\theta -\Eta)P=:V_3$              \\
$(P,0,0)$                                                   & $-\Lambda P=:V_4$                              \\
$(0,uP,0)$                                                  & $(\theta -\Eta)uP=:V_5$                       \\
$\left(\frac{u}{1+u}P,\frac{u}{1+u}P,\frac{u}{1+u}P\right)$ & $(v-\Lambda+\theta -\Eta)\frac{u}{1+u}P=:V_6$ \\ \bottomrule
\end{tabular}
\caption{Extreme points and corresponding values of the objective function}
\label{tab:extreme-and-value}
\end{table}
\endgroup

\paragraph{Case 1: $\Eta\leq \theta ,\ \Lambda\leq v$.}\hspace{0em}

In this case, $v-\Lambda,\ \theta -\Eta\geq0$. For any $a\in[0,P], b\in[0,1]$, $(v-\Lambda)a+(\theta -\Eta)b\leq (v-\Lambda)P+\theta -\Eta=V_1$. Therefore, $V_1\geq V_2,V_3,V_5,V_6\geq 0\geq V_4$, and the solution is $(p(\theta),Q(\theta))=(P,1)$.

\paragraph{Case 2: $\Eta\leq \theta ,\ \Lambda\geq v$.}\hspace{0em}

In this case, $v-\Lambda\leq 0,\ \theta -\Eta\geq0$. For any $a\in[0,P], b\in[0,1]$, $(v-\Lambda)a+(\theta -\Eta)b\leq \theta -\Eta=V_2$. Therefore, $V_2\geq V_1,V_3,V_5,V_6$. Also, $V_2\geq 0\geq V_4$, and hence the solution is $(p(\theta),Q(\theta))=(0,1)$.

\paragraph{Case 3: $\Eta\geq \theta ,\ \Lambda+\Eta\leq \theta +v$.}\hspace{0em}

In this case, it is easy to verify that $v-\Lambda\geq 0$, $\theta -\Eta\leq 0$, and $v-\Lambda+\theta -\Eta\geq 0$. Therefore, $V_2, V_4, V_5\leq 0\leq V_3$. Furthermore, 
\begin{equation*}
    V_1=(v-\Lambda)P+\theta -\Eta\leq (v-\Lambda)P+(\theta -\Eta)P=V_3,
\end{equation*}
and
\begin{equation*}
    V_6 = (v-\Lambda+\theta -\Eta)\frac u{1+u}P\leq (v-\Lambda+\theta -\Eta)P=V_3.
\end{equation*}
Hence, $V_3$ is the largest, and the solution is $(p(\theta),Q(\theta))=(P,P)$.

\paragraph{Case 4: $\Eta\geq \theta ,\ \Eta-\frac \Lambda u\geq \theta +v$.}\hspace{0em}

In this case, it is easy to verify that $\theta -\Eta\leq 0$ and $v-\Lambda+\theta -\Eta\leq 0$. Therefore, $V_1\leq V_3\leq V_6\leq 0$, $V_2\leq V_5\leq 0$.

Since $\Eta-\frac \Lambda u\geq \theta +v$, we know that $\Eta-\frac \Lambda u- \theta \geq v$, and that
\begin{equation*}
    V_4-V_5=\left(\Eta-\theta -\frac\Lambda u\right)uP\geq vuP\geq 0\implies  V_4\geq V_5.
\end{equation*}
Furthermore, $\Eta-\frac \Lambda u- \theta - v\geq 0$, and
\begin{align*}
    V_4-V_6&=\left(-\frac{1+u}{u}\Lambda - v+\Lambda-\theta +\Eta\right)\frac u{1+u}P\\&=\left(\Eta-\frac\Lambda u-\theta -v \right)\frac{u}{1+u}P\\&\geq 0 \implies V_4\geq V_6.
\end{align*}

Hence, $V_4$ is the largest, and the solution is $(p(\theta),Q(\theta))=(P,0)$.

\paragraph{Case 5: $\Eta\geq \theta ,\ \Eta-\frac\Lambda u\leq \theta -\frac vu$.}\hspace{0em}

In this case, it is easy to verify that $v-\Lambda\leq 0$ and $\theta -\Eta\leq 0$. Therefore, $V_3\leq V_6\leq 0$, $V_1\leq V_2\leq V_5\leq 0$.

Since $\Eta-\frac\Lambda u\leq \theta -\frac vu$ and therefore $\theta -\Eta+\frac\Lambda u\geq \frac vu$, we have
\begin{equation*}
    V_5-V_4=\left(\theta -\Eta+\frac\Lambda u\right)uP \geq \frac vu uP \geq 0\implies V_5\geq V_4,
\end{equation*}
and 
\begin{align*}
    V_5-V_6&=[(1+u)(\theta -\Eta)-(v-\Lambda+\theta -\Eta)]\frac{u}{1+u}P\\
    &=\left(\theta -\Eta+\frac \Lambda u-\frac vu\right)\frac{u^2}{1+u}P\\
    &\geq \left(\frac vu-\frac vu\right)\frac{u^2}{1+u}P\\
    &=0\implies V_5\geq V_6.
\end{align*}

Hence, $V_5$ is the largest, and the solution is $(p(\theta),Q(\theta))=(0,uP)$.

\paragraph{Case 6: $\Eta\geq \theta ,\ \Lambda+\Eta\geq \theta +v,\ \theta -\frac vu\leq \Eta-\frac \Lambda u\leq \theta +v$.}\hspace{0em}

In this case, it is easy to verify that $\theta -\Eta\leq 0$ and $v-\Lambda+\theta -\Eta\leq 0$. Therefore, $V_1\leq V_3\leq V_6\leq 0$, $V_2\leq V_5\leq 0$.

$\theta -\frac vu\leq \Eta-\frac \Lambda u\leq \theta +v$ implies $\theta -\Eta+v+\frac \Lambda u\geq 0$ and $\frac vu-\frac \Lambda u-\theta +\Eta\geq 0$. Therefore, 
\begin{align*}
    V_6-V_4&=\left(v-\Lambda+\theta -\Eta+\frac{1+u}{u}\Lambda\right)\frac{u}{1+u}P\\
    &=\left(\theta -\Eta+v+\frac{\Lambda}{u}\right)\frac{u}{1+u}P\\
    &\geq 0\implies V_6\geq V_4,
\end{align*}
and 
\begin{align*}
    V_6-V_5&=[v-\Lambda+\theta -\Eta-(1+u)(\theta -\Eta)]\frac{u}{1+u}P\\
    &=\left(\frac vu-\frac \Lambda u-\theta +\Eta\right)\frac{u^2}{1+u}P\\
    &\geq 0\implies V_6\geq V_5.
\end{align*}

Hence, $V_6$ is the largest, and the solution is $(p(\theta),Q(\theta))=(\frac u{1+u}P,\frac u{1+u}P)$.

\vspace{1em}
In summary, the solution to problem \eqref{eq:lagrangian} is as listed in \eqref{eq:beautiful-graph-sol}.
\end{proof}

\begin{figure}[htbp]
  \centering
  \begin{subfigure}{0.49\textwidth}
    \centering
\begin{tikzpicture}[scale=.6, xscale=9, yscale=6, every node/.style={font=\small, scale=.7}]
\path[use as bounding box] (0,-.075) rectangle (.75,1.575);
\def\v{0.4}        
\def\u{.6}        
\def\Finv{.5}       
\def\mumax{1.5}    

\coordinate (O)  at (0,0);
\coordinate (A)  at (0,\Finv);       
\coordinate (Av) at (\v,\Finv);      
\coordinate (L)  at (\v,0);          
\coordinate (R)  at (1,0);           
\coordinate (R1) at (1,\Finv);           
\coordinate (TopLeft) at (0,\mumax);
\coordinate (TopRight) at (1,\mumax);

\coordinate (P1) at (0,\Finv+\v);
\coordinate (P2) at (\mumax*\u-\Finv*\u-\v*\u,\mumax);         
\coordinate (P3) at (\mumax*\u-\Finv*\u+\v,\mumax);

\fill[red!25] (O) -- (L) -- (Av) -- (A) -- cycle;

\fill[violet!20] (L) -- (R) -- (R1) -- (Av) -- cycle;

\fill[orange!35] (A) -- (P1) -- (Av) -- cycle;

\fill[yellow!35] (P1) -- (TopLeft) -- (P2) -- cycle;

\fill[blue!20] (Av) -- (R1) -- (TopRight) -- (P3) -- cycle;

\fill[teal!20] (P1) -- (Av) -- (P3) -- (P2) -- cycle;

\draw[->, >=stealth] (-0.02,0) -- (1.05,0) node[below right]{$\Lambda$};
\draw[->, >=stealth] (0,-0.02) -- (0,\mumax+0.05) node[left]{$\Eta$};

\draw[line width=1pt, gray!70] (P1) -- (Av) node[midway, font=\scriptsize, rotate=-33.7, above]{$\text{slope}=-1$};

\draw[line width=1pt, gray!70] (P1) -- (P2)
node[midway, font=\scriptsize, rotate=48, below]{$\text{slope}=u^{-1}$};

\draw[line width=1pt, gray!70] (Av) -- (P3) node[midway, font=\scriptsize, rotate=48, below]{$\text{slope}=u^{-1}$};

\draw[line width=1pt, gray!70] (A) -- (R1);
\draw[line width=1pt, gray!70] (L) -- (Av);

\draw (0,\Finv) node[left]{$\theta\ $} -- +(-0.02,0);
\draw (0,\Finv+\v) node[left]{$\theta+v\ $} -- +(-0.02,0);
\draw (\v,0) node[below]{$v$} -- +(0,0.02);
\node[below left] at (O) {$0$};

\node[red!60!black,font=\Large] at (\v/2,\Finv/2) {$(P,1)$};
\node[violet!60!black,font=\Large] at (\v/2+0.5,\Finv/2) {$(0,1)$};
\node[orange!50!black,font=\Large] at (\v/3,\Finv+\v/4) {$(P,P)$};
\node[yellow!30!red!60!black,font=\Large] at (\v/3.8,\Finv+\v*2) {$(P,0)$};
\node[blue!80!black,font=\Large] at (\v/1.5+0.5,\Finv+\v/2) {$(0,uP)$};
\node[teal!70!black,font=\Large] at (\v*.9,\v*1.1+\Finv) {$\left(\tfrac{u}{1+u}P,\tfrac{u}{1+u}P\right)$};
\end{tikzpicture}\hspace{35pt}
\caption{}
\label{fig:beautiful-graph-0}
  \end{subfigure}
  \hfill
\begin{subfigure}{0.49\textwidth}
    \centering
\begin{tikzpicture}[scale=.6, xscale=9, yscale=4.5, every node/.style={font=\small, scale=.7}]
\path[use as bounding box] (0,-.1) rectangle (.7,2.1);
\def\v{0.4}        
\def\u{.6}        
\def\Finv{1}       
\def\mumax{2}    

\coordinate (O)  at (0,0);
\coordinate (A)  at (0,\Finv);       
\coordinate (Av) at (\v,\Finv);      
\coordinate (L)  at (\v,0);          
\coordinate (R)  at (1,0);           
\coordinate (R1) at (1,\Finv);           
\coordinate (TopLeft) at (0,\mumax);
\coordinate (TopRight) at (1,\mumax);

\coordinate (P1) at (0,\Finv+\v);
\coordinate (P2) at (\mumax*\u-\Finv*\u-\v*\u,\mumax);         
\coordinate (P3) at (\mumax*\u-\Finv*\u+\v,\mumax);

\begin{scope}[opacity=0.6]
\fill[red!25] (O) -- (L) -- (Av) -- (A) -- cycle;

\fill[violet!20] (L) -- (R) -- (R1) -- (Av) -- cycle;

\fill[orange!35] (A) -- (P1) -- (Av) -- cycle;

\fill[yellow!35] (P1) -- (TopLeft) -- (P2) -- cycle;

\fill[blue!20] (Av) -- (R1) -- (TopRight) -- (P3) -- cycle;

\fill[teal!20] (P1) -- (Av) -- (P3) -- (P2) -- cycle;
\end{scope}

\draw[->, >=stealth] (-0.02,0) -- (1.05,0) node[below right]{$\Lambda$};
\draw[->, >=stealth] (0,-0.02) -- (0,\mumax+0.05) node[left]{$\Eta+1-\theta$};

\draw[line width=1pt, gray!30] (P1) -- (Av) node[midway, font=\scriptsize, fill opacity=0.5, rotate=-26.565, above]{};

\draw[line width=1pt, gray!30] (P1) -- (P2)
node[midway, font=\scriptsize, fill opacity=0.5, rotate=39.8, below]{};

\draw[line width=1pt, gray!30] (Av) -- (P3) node[midway, font=\scriptsize, fill opacity=0.5, rotate=39.8, below]{};

\draw[line width=1pt, gray!30] (A) -- (R1);
\draw[line width=1pt, gray!30] (L) -- (Av);

\draw (0,\Finv) node[left]{$1\ $} -- +(-0.02,0);
\draw (0,\Finv+\v) node[left]{$1+v\ $} -- +(-0.02,0);
\draw (\v,0) node[below]{$v$} -- +(0,0.02);
\node[below left] at (O) {$0$};

\begin{scope}[opacity=0.2]
\node[red!60!black!, font=\Large] at (\v/2,\Finv/2) {$(P,1)$};
\node[violet!60!black, font=\Large] at (\v/2+0.5,\Finv/2) {$(0,1)$};
\node[orange!50!black, font=\Large] at (\v/3,\Finv+\v/4) {$(P,P)$};
\node[yellow!30!red!60!black, font=\Large] at (\v/3.8,\Finv+\v*2) {$(P,0)$};
\node[blue!80!black, font=\Large] at (\v/1.5+0.5,\Finv+\v/2) {$(0,uP)$};
\node[teal!70!black, font=\Large] at (\v*.9,\v*1.1+\Finv) {$\left(\tfrac{u}{1+u}P,\tfrac{u}{1+u}P\right)$};
\end{scope}

\def\mechmu{0.8}
\coordinate (ONE1)  at (\v/3,\mechmu);
\coordinate (TOP1)  at (\v/3,1);
\coordinate (MID1)  at (\v/3,1+\v*2/3);
\coordinate (BOT1)  at (\v/3,1+\v+\v/\u/3);
\coordinate (ZERO1)  at (\v/3,\mechmu+1);
\coordinate (HALFWAY1)  at (\v/3,\mechmu+.3);


\draw[line width=1.5pt, black] (ONE1) -- (ZERO1);

\node[circle, fill=black, inner sep=1.5pt, label=below:$(\Lambda\comma\Eta)$] at (ONE1) {};
\node[circle, fill=black, inner sep=1.5pt, label=above right:$\hspace{-20pt}(\Lambda\comma\Eta+1)$] at (ZERO1) {};
\node[circle, fill=black, inner sep=1.5pt] at (HALFWAY1) {};

\coordinate (ARRFROM)  at (\v/3+0.05+0.15,1+\v+0.3);
\coordinate (ARRTO)  at (\v/3+0.05,1+\v+.1);


\draw [decorate,black,decoration={brace,amplitude=6pt, mirror},line width=.5pt] 
    (ZERO1) -- (HALFWAY1) 
    node[midway, left] {$\theta\hspace{8pt}$};

\coordinate (ONE2) at (\v+.2,.6);
\coordinate (ZERO2) at (\v+.2,1.6);

\draw[line width=1.5pt, black] (ONE2) -- (ZERO2);

\node[circle, fill=black, inner sep=1.5pt] at (ONE2) {};
\node[circle, fill=black, inner sep=1.5pt] at (ZERO2) {};

\end{tikzpicture}\hspace{45pt}
\caption{}
\label{fig:beautiful-graph-1}
  \end{subfigure}
  \caption{Illustration of direct mechanisms.}
  \label{fig:beautiful-graph}
\end{figure}
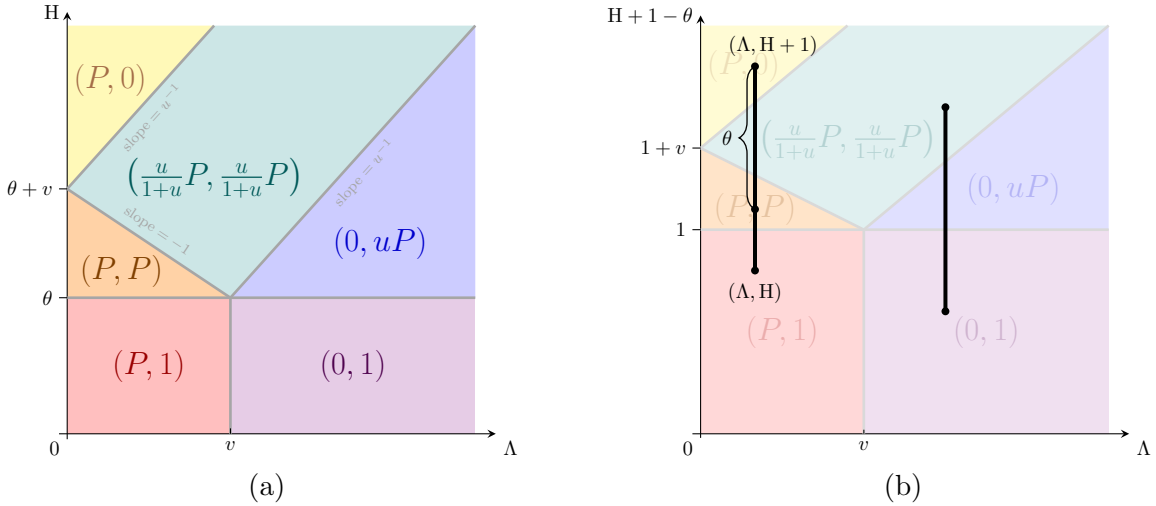

\Cref{fig:beautiful-graph-0} provides an illustration for the results in \Cref{prop:beautiful-graph} in $(\Lambda,\Eta)$ space.

The second step, finding a combination of $\lambda$ and $\eta$ such that $(p_{\lambda,\eta},Q_{\lambda,\eta})$ satisfies the capacity and obedience constraint, is not useful for characterizing the optimal mechanism.  

By \Cref{prop:beautiful-graph}, there are two classes of mechanisms that are potentially optimal, illustrated by the two line segments in \Cref{fig:beautiful-graph-1}. When $\Lambda\leq v$, the solution to problem \eqref{eq:lagrangian} is a four-segment mechanism that takes the following form:
\begin{equation}
\label{eq:solultion-pq}
    (p(\theta),Q(\theta))=\switch{
        &(P,1), &\theta_3<\theta&\leq1,\\
        &(P,P), &\theta_2<\theta&\leq\theta_3,\\
        &\left(\frac u{1+u}P,\frac u{1+u}P\right), &\theta_1<\theta&\leq\theta_2,\\
        &(P,0), &0\leq\theta&\leq\theta_1,
    }
\end{equation}
where $P\in[\kappa,1]$, and
\begin{align*}
    \theta_1&=\Eta-v-\frac\Lambda u,\\
    \theta_2&=\Eta-v+\Lambda,\\
    \theta_3&=\Eta.
\end{align*}
Each segment could be empty. This is the class of mechanisms that the optimal direct mechanism described in \Cref{thm:optimal-direct} falls into. When $\Lambda>v$, the solution to problem \eqref{eq:lagrangian} is a mechanism that takes another four-segment form:
\begin{equation*}
    (p(\theta),Q(\theta))=\switch{
        &(0,1), &\theta_3<\theta&\leq1,\\
        &(0,uP), &\tilde\theta_2<\theta&\leq\theta_3,\\
        &\left(\frac u{1+u}P,\frac u{1+u}P\right), &\theta_1<\theta&\leq\tilde\theta_2,\\
        &(P,0), &0\leq\theta&\leq\theta_1,
    }
\end{equation*}
where $P\in[\kappa,1]$, and
\begin{align*}
    \tilde\theta_2&=\Eta+\frac vu-\frac\Lambda u.
\end{align*}
Each segment could be empty. \Cref{lem:no-weird-case} shows that it is sufficient to consider the first class of mechanisms.

\begin{lemma}
    \label{lem:no-weird-case}
    Let $(p,Q)$ represent the optimal direct mechanism that solves the sender's problem \eqref{eq:sender-prob}. There exists $\lambda,\eta$ satisfying $\Lambda=\lambda/\eta\leq v$ such that $(p,Q)$ solves problem \eqref{eq:lagrangian} given $P=\sup_{\htt\in\Theta}p(\htt)$.
\end{lemma}

\begin{proof}
    To prove this, it is sufficient to prove the following claim.

    \begin{claim*}
        For a given $P$, let $\pi=(p,Q)$ be the mechanism that solves problem \eqref{eq:fix-P-max-over-p-Q}. If all corresponding pairs of multipliers $(\lambda,\eta)$ satisfy $\Lambda:=\lambda/\eta>v$ (i.e., $(p,Q)$ falls only into the second class), then there exists another feasible mechanism $(p',Q')$ that strictly dominates $(p,Q)$ and has corresponding Lagrangian multipliers $\lambda',\eta'$ such that $\Lambda':=\lambda'/\eta'\leq v$.
    \end{claim*}
    We move on to prove this claim. When $\Lambda>v$, by \Cref{prop:beautiful-graph}, $(p,Q)$ must take the following form:
    \begin{equation*}
    (p(\theta),Q(\theta))=\switch{
        &(0,1), &\theta_3<\theta&\leq1,\\
        &(0,uP), &\tilde\theta_2<\theta&\leq\theta_3,\\
        &\left(\frac u{1+u}P,\frac u{1+u}P\right), &\theta_1<\theta&\leq\tilde\theta_2,\\
        &(P,0), &0\leq\theta&\leq\theta_1,
    }
\end{equation*}
    We assume $0<\theta_1<\tilde\theta_2<1$, i.e., the fourth segment and one of the first two segments is of strictly positive measure. This is without loss of generality because of the following reasons: when the first two segments are measure zero, the mechanism also falls into the first class; when the fourth segment is measure zero, the mechanism is equivalent to a mechanism with $p(\theta)=0$ in the fourth segment, in which $\sup_{\theta\in\Theta} p(\theta)=\frac{u}{1+u}P<P$, and is thus weakly dominated by the optimal mechanism of \eqref{eq:fix-P-max-over-p-Q} when $P$ is replaced by $\frac u{1+u}P$. Since $F$ has positive density almost everywhere, this assumption implies $F(\theta_1)>0$, $F(\tilde\theta_2)<1$, and $\int_\Theta Q(\theta)\D F(\theta)>0$.
    
    Now, consider the following mechanism $\tilde\pi'=(\tilde{p}',\tilde{Q}')$:
    \begin{equation*}
        \tilde{p}'(\theta)\equiv\kappa,\ \tilde{Q}'(\theta)=Q(\theta),\ \forall \theta\in\Theta.
    \end{equation*}
    
    It is easy to verify that $(\tilde{p}',\tilde{Q}')$ also satisfies capacity and IC constraints in problem \eqref{eq:sender-prob}, with $\sup_{\theta\in\Theta}p(\theta)=\kappa$, and yields exactly the same amount of recommendation as $(p,Q)$.

    Consider the obedience constraint. When $\kappa>\frac{u}{1+u}P$, it is easy to verify that $\min\{\tilde p'(\theta),\tilde Q'(\theta)\}\geq \min\{p(\theta),Q(\theta)\}$ for all $\theta\in\Theta$, and the inequality is strict when $\theta\in(\tilde\theta_2,1]$ which is of strictly positive measure. Therefore, we have
    \begin{align*}
        \int_\Theta \min\{\tilde p'(\theta),\tilde Q'(\theta)\}\D F(\theta)>\int_\Theta \min\{p(\theta),Q(\theta)\}\D F(\theta).
    \end{align*}
    
    When $\kappa\leq \frac{u}{1+u}P$, the capacity constraint gives
    \begin{align*}
        \int_\Theta p(\theta)\D F(\theta)\leq \kappa\implies PF(\theta_1)+\frac u{1+u}P(F(\tilde\theta_2)-F(\theta_1))\leq \kappa.
    \end{align*}
    This implies
    \begin{align*}
        &\int_\Theta \min\{\tilde p'(\theta),\tilde Q'(\theta)\}\D F(\theta)-\int_\Theta \min\{p(\theta),Q(\theta)\}\D F(\theta)\\
        =~&\kappa (1-F(\theta_1))-\frac{u}{1+u}P(F(\tilde\theta_2)-F(\theta_1))\\
        \geq~&\left[PF(\theta_1)+\frac u{1+u}P(F(\tilde\theta_2)-F(\theta_1))\right](1-F(\theta_1))-\frac{u}{1+u}P(F(\tilde\theta_2)-F(\theta_1))\\
        =~&PF(\theta_1)\left(1-\frac u{1+u}F(\tilde\theta_2)-\frac 1{1+u}F(\theta_1)\right)\\
        \geq~&PF(\theta_1)(1-F(\tilde\theta_2))\\
        >~&0.
    \end{align*}
    Hence,
    \begin{align*}
        \int_\Theta \min\{\tilde p'(\theta),\tilde Q'(\theta)\}\D F(\theta)>\int_\Theta \min\{p(\theta),Q(\theta)\}\D F(\theta)
    \end{align*}
    holds for both $\kappa>\frac u{1+u}P$ and $\kappa\leq\frac u{1+u}P$. Since $\tilde Q'=Q$, this further indicates that
    \begin{align*}
        &\mu_{\tilde\pi'}=\frac{\displaystyle\int_\Theta \theta \tilde Q'(\theta)\D F(\theta)+v\int_\Theta \min\{\tilde p'(\theta),\tilde Q'(\theta)\}\D F(\theta)}{\displaystyle\int_\Theta\tilde Q'(\theta) \D F(\theta)}\\
        >~&\mu_\pi=\frac{\displaystyle\int_\Theta \theta Q(\theta)\D F(\theta)+v\int_\Theta \min\{ p(\theta),Q(\theta)\}\D F(\theta)}{\displaystyle\int_\Theta Q(\theta) \D F(\theta)}\\
        \geq~&\utt,
    \end{align*}
    i.e., the obedience constraint of $(\tilde p',\tilde Q')$ is strictly slack. Therefore, it is feasible and must be strictly dominated by the optimal solution to \eqref{eq:fix-P-max-over-p-Q} given $P=\kappa$, which we denote by $(p',Q')$. 

    Now, $\sup_{\theta\in\Theta}p'(\theta)=\kappa$ implies that there exists a pair of corresponding multipliers $(\lambda',\eta')$ satisfying $\Lambda'=\lambda'/\eta'=0<v$. Also, $(p',Q')$ strictly dominates $(\tilde p',\tilde Q')$ and therefore $(p,Q)$ since $\tilde Q'=Q$. This proves the claim and hence the lemma.
\end{proof}

A remark here is that for a given $P$, the direct mechanism that solves problem \eqref{eq:fix-P-max-over-p-Q} may still correspond to a $\Lambda>v$. This is a sign that the supremum of $p(\theta)$, $P$, is chosen to be too large. 



\Cref{lem:no-weird-case} implies that the optimal direct mechanism must take the form of \eqref{eq:solultion-pq}. Applying the result of \Cref{lem:A-is-minpq} to recover the recommendation rule $q$ from the ex-ante probability of recommendation $Q$ in \eqref{eq:solultion-pq} yields the exact four-segment form of the optimal direct mechanism as presented in \Cref{thm:optimal-direct}. The last step is to exclude the case where $\sup_{\theta\in\Theta}p(\theta)=\kappa$.

\begin{proposition}
    \label{prop:indirect-not-optimal}
    If an optimal mechanism to \eqref{eq:sender-prob}, $(p^*,Q^*)$, satisfies $\int_\Theta Q^*(\theta)\D F(\theta)>0$, then $\sup_{\theta\in\Theta}p^*(\theta)>\kappa$.
\end{proposition}

\begin{proof}
    When $p(\theta)\equiv \kappa$, $\sup_{\theta\in\Theta}p(\theta)=\kappa$, and when $P=\kappa$, the optimal $p$ in problem \eqref{eq:fix-P-max-over-p-Q} is $p(\theta)\equiv \kappa$. We go on to prove that any mechanism with $p(\theta)\equiv\kappa$ is always improvable.
    
    According to \Cref{prop:beautiful-graph}, the solution to problem \eqref{eq:fix-P-max-over-p-Q} given $P=\kappa$ takes the following form:
    \begin{equation*}
        p(\theta)\equiv \kappa ,\ Q(\theta)=\switch{
        &1, &\Eta&<\theta\leq 1,\\
        &\kappa, &\Eta-v&<\theta\leq \Eta,\\
        &0, &0&\leq\theta\leq \Eta-v,
        }
    \end{equation*}
    where $\Eta=\utt-\eta^{-1}$ and $(\lambda,\eta)$ are the corresponding multipliers satisfying $\Lambda=\lambda/\eta=0$. 
    
    Denote by $(p^*,Q^*)$ the optimal solution to the sender's problem \eqref{eq:sender-prob}. The condition $\int_\Theta Q^*(\theta)\D F(\theta)>0$ implies that at least someone is recommended in the best case. If $Q(\theta)=0$ almost everywhere, it is clearly strictly dominated by $(p^*,Q^*)$ which must correspond to $P>\kappa$.
    
    When $\int_\Theta Q(\theta)\D F(\theta)>0$, given that the obedience constraint binds, $Q$ cannot be equal to $1$ almost everywhere, and hence the second segment where $Q(\theta)=\kappa$ must be of strictly positive measure.

    Problem \eqref{eq:lagrangian} at $\Lambda=0$ becomes
    \begin{align*}
        \max_{p,Q}\ &\int_\Theta [(\theta-\Eta) Q(\theta) + v \min\{p(\theta),Q(\theta)\}]\D F(\theta),\\
        \st&\text{(IC$_1$) } Q(\theta)+p(\theta)u\geq \kappa u,\ \forall \theta\in\Theta,\\
        &\text{(IC$_2$) } p(\theta)\leq \kappa,\ \forall \theta\in\Theta.
    \end{align*}
    Here, we scaled the objective function by $\eta^{-1}$. We introduce Lagrangian multipliers, $\nu(\theta),\ \phi(\theta)\geq 0$, for constraint (IC$_1$) and (IC$_2$) respectively.

    For a given $\theta$, the pointwise optimization problem is 
    \begin{equation*}
    \begin{aligned}
        \max_{p(\theta),Q(\theta)}\ &(\theta -\Eta)Q(\theta) + v \min\{p(\theta),Q(\theta)\},\\
        \st&\text{(IC$_1$) } Q(\theta)+p(\theta)u\geq \kappa u,\\
        &\text{(IC$_2$) } p(\theta)\leq \kappa,\\
        &Q(\theta)\leq 1.
    \end{aligned}
    \end{equation*}
     We dropped the constraints $p(\theta),Q(\theta)\geq 0$ because it is never optimal to decrease $p(\theta)$, and $Q(\theta)\geq0$ is implied by the two IC constraints. We further introduce $\psi(\theta)$ to be the multiplier corresponding to $Q(\theta)\leq 1$.

     When $\theta\in(\Eta,1]$, the optimal $(p(\theta),Q(\theta))=(\kappa,1)$ suggests that (IC$_1$) is slack, i.e. $\nu(\theta)=0$, and $\min\{p(\theta),Q(\theta)\}=p(\theta)$. The Lagrangian problem is thus
     \begin{equation*}
         \max_{p(\theta),Q(\theta)}(\theta -\Eta -\psi(\theta))Q(\theta)+vp(\theta)-\phi(\theta)(p(\theta)-\kappa)+\psi(\theta).
     \end{equation*}
     First order condition gives $\phi(\theta)=v>0$.

     When $\theta\in(\Eta-v,\Eta]$, the optimal $(p(\theta),Q(\theta))=(\kappa,\kappa)$ suggests that (IC$_1$) and $Q(\theta)\leq 1$ are slack, i.e. $\nu(\theta)=\psi(\theta)=0$. We write the problem equivalently as
    \begin{equation*}
    \begin{aligned}
        \max_{p(\theta),Q(\theta),m}\ &(\theta -\Eta)Q(\theta) + v\cdot m,\\
        \st\ & p(\theta)\leq \kappa,\\
        & p(\theta)\geq m,\\
        & Q(\theta)\geq m.
    \end{aligned}
    \end{equation*}
    Let $\alpha,\beta\geq 0$ be the Lagrangian multipliers of $p(\theta)\geq m$ and $Q(\theta)\geq m$ respectively. The corresponding Lagrangian problem is
    \begin{equation*}
        \max_{p(\theta),Q(\theta)} (\theta -\Eta)Q(\theta)+v\cdot m-\phi(\theta)(p(\theta)-\kappa)-\alpha(m-p(\theta))-\beta(m-Q(\theta)).
    \end{equation*}
    First order conditions give
    \begin{align}
        \theta -\Eta+\beta&=0,\tag{FOC-$Q(\theta)$}\\
        -\phi(\theta)+\alpha&=0,\tag{FOC-$p(\theta)$}\\
        v-\alpha-\beta&=0.\tag{FOC-$m$}
    \end{align}
    These conditions yield
    \begin{equation*}
        \phi(\theta)=\alpha=\theta -\Eta+v> 0,\ \beta=\Eta-\theta \geq 0.
    \end{equation*}

    When $\theta\in[0,\Eta-v]$, we know from the optimal $(p(\theta),Q(\theta))=(\kappa,0)$ that $\min\{p(\theta),Q(\theta)\}=Q(\theta)$, and that $Q(\theta)\leq 1$ is slack, i.e. $\psi(\theta)=0$. The corresponding Lagrangian problem is thus
    \begin{equation*}
        \max_{p(\theta),Q(\theta)} (\theta -\Eta+v)Q(\theta)-\nu(\theta)(\kappa u-Q(\theta)-p(\theta)u)-\phi(\theta)(p(\theta)-\kappa).
    \end{equation*}
    First order condition of $p(\theta)$ gives
    \begin{align*}
        \nu(\theta)u-\phi(\theta)=0.
    \end{align*}

    In summary, at $P=\kappa$, the Lagrangian multipliers for incentive compatibility constraints in problem \eqref{eq:lagrangian} and therefore in problem \eqref{eq:fix-P-max-over-p-Q} satisfy
    $\nu(\theta)=0$, $\phi(\theta)>0$ for all $\theta\in(\Eta-v,1]\neq \varnothing$, and $\nu(\theta)u-\phi(\theta)=0$ for all $\theta\in[0,\Eta-v]$.

    Denote the optimal value of problem \eqref{eq:fix-P-max-over-p-Q} be $V$. By envelope theorem, 
    \begin{align*}
        \left.\frac{\partial V}{\partial P}\right|_{P=\kappa^+}&=\eta\int_\Theta[-\nu(\theta)u+\phi(\theta)]\D F(\theta)\\
        &=\eta\int_{\theta\in(\Eta-v,1]}[-\nu(\theta)u+\phi(\theta)]\D F(\theta)~+~\eta\int_{\theta\in[0,\Eta-v]}[-\nu(\theta)u+\phi(\theta)]\D F(\theta)\\
        &=\eta\int_{\theta\in(\Eta-v,1]}\phi(\theta)\D F(\theta)~+~0\\
        &>0.
    \end{align*}
    Hence, when $P=\kappa$, the sender will be strictly better off by increasing $P$, and thus $p(\theta)\equiv \kappa$ cannot be supported in any optimal direct mechanism.
\end{proof}

\Cref{prop:indirect-not-optimal} restricts the range of $P$ in \eqref{eq:solultion-pq} to $(\kappa,1]$ whenever the sender can recommend a strictly positive measure of agents in the best situation, which finishes the proof of \Cref{thm:optimal-direct}. \Cref{cor:joint-improves-separated} follows directly from \Cref{prop:indirect-not-optimal}.

\end{appendices}

\end{document}